\documentclass{lmcs}
\pdfoutput=1

\usepackage{lastpage}
\lmcsdoi{15}{2}{6}
\lmcsheading{}{\pageref{LastPage}}{}{}%
{Jun.~22,~2017}{Apr.~30,~2019}{}

\usepackage[utf8]{inputenc}
\makeatletter
\let\cc\@undefined
\let\fact\@undefined
\let\endfact\@undefined
\makeatother
\usepackage{complexity}
\usepackage{amsmath}
\usepackage{amssymb}
\DeclareSymbolFont{Shuffle}{U}{shuffle}{m}{n}
\DeclareFontFamily{U}{shuffle}{}
\DeclareFontShape{U}{shuffle}{m}{n}{%
  <-8>shuffle7%
  <8->shuffle10%
}{}
\DeclareMathSymbol\shuffle{\mathbin}{Shuffle}{"001}
\DeclareMathSymbol\cshuffle{\mathbin}{Shuffle}{"002}
\usepackage{mathtools}
\usepackage{amsthm}
\usepackage{mathabx}
\usepackage{cleveref}
\usepackage{url}
\usepackage{thmtools}
\theoremstyle{plain}
\newtheorem{theorem}{Theorem}[section]
\newtheorem{lemma}[theorem]{Lemma}
\newtheorem{proposition}[theorem]{Proposition}
\newtheorem{corollary}[theorem]{Corollary}
\newtheorem{fact}[theorem]{Fact}
\theoremstyle{definition}
\newtheorem{definition}[theorem]{Definition}

\theoremstyle{remark}
\newtheorem{remark}[theorem]{Remark}
\newtheorem{claim}[theorem]{Claim}

\newcommand{\subA}{{\subword_{A^*}}}
\newcommand{\nsubA}{{\not\subword_{A^*}}}
\newcommand{\ssubA}{{\ssubword_{A^*}}}
\newcommand{\nssubA}{{\not\ssubword_{A^*}}}
\newcommand{\supA}{{\supword_{A^*}}}
\newcommand{\nsupA}{{\not\supword_{A^*}}}
\newcommand{\ssupA}{{\ssupword_{A^*}}}
\newcommand{\perpA}{{\perp_{A^*}}}
\newcommand{\idA}{{=_{A^*}}}
\newcommand{\etal}{\textit{et al}. }
\newcommand{\stc}{\textit{sc}}

\newcommand{\equivdef}{\stackrel{\mbox{\begin{scriptsize}def\end{scriptsize}}}{\iff}}
\newcommand{\egdef}{\stackrel{\mbox{\begin{scriptsize}def\end{scriptsize}}}{=}}

\newcommand{\subword}{\sqsubseteq}
\newcommand{\supword}{\sqsupseteq}
\newcommand{\ssubword}{\sqsubset}
\newcommand{\ssupword}{\sqsupset}
\newcommand{\strictsubword}{\ssubword}

\newcommand{\step}[1]{\xrightarrow{\!\!#1\!\!}}
\newcommand{\calA}{\mathcal A}
\newcommand{\calB}{\mathcal B}
\newcommand{\calJ}{\mathcal J}

\newcommand{\Nat}{\mathbb{N}}
\newcommand{\Reals}{\mathbb{R}}
\newcommand{\up}{{\uparrow}}
\newcommand{\down}{{\downarrow}}
\newcommand{\simonRstar}{\lesssim_n}

\newcommand{\obracew}[2]{{\overset{#2}{\overbrace{#1}}}}
\newcommand{\DD}[1]{\down{#1}}
\newcommand{\LL}[1]{\up{#1}}
\newcommand{\LLA}[1]{\up_A{#1}}
\newcommand{\pair}[2]{\begin{bmatrix} #2 \\ #1 \end{bmatrix}}
\def\FO{\ComplexityFont{FO}}
\def\Rat{\ComplexityFont{Rat}}
\def\Rec{\ComplexityFont{Rec}}
\def\Reg{\ComplexityFont{Reg}}
\def\EXPTIME{\ComplexityFont{EXPTIME}}

\def\EXPTIMEiii{\ComplexityFont{3-}\EXPTIME}
\renewcommand{\setminus}{\smallsetminus} 
\renewcommand{\epsilon}{\varepsilon}     

\begin{document}
\title[The height of piecewise-testable languages]{The height of piecewise-testable languages and the complexity of the logic of subwords}
\author[P. Karandikar]{Prateek Karandikar}
\address{CMI, Chennai, India}
\email{prateek@cmi.ac.in}
\author[Ph.~Schnoebelen]{Philippe Schnoebelen}
\address{LSV, CNRS \& ENS Paris-Saclay, France}
\email{phs@lsv.fr}
\thanks{Work partially supported by ANR grant ANR-14-CE28-0002 PACS. The first
author is now with Google, London.}
\titlecomment{An extended abstract of this work first appeared in the
  proceedings of the 25th EACSL Conference on Computer Science Logic
   (CSL 2016)~\cite{KS-csl2016}.
}

\begin{abstract}
The height of a piecewise-testable language $L$ is the maximum length
of the words needed to define $L$ by excluding and requiring given
subwords.  The height of $L$ is an important descriptive complexity
measure that has not yet been investigated in a systematic way.
This article develops a series of new techniques for bounding the
height of finite languages and of languages obtained by taking
closures by subwords, superwords and related operations.

As an application of these results, we show that
$\FO^2(A^*,\subword)$, the two-variable fragment of the first-order
logic of sequences with the subword ordering, can only express
piecewise-testable properties and has elementary complexity.
\end{abstract}


\maketitle


\section{Introduction}
\label{sec-intro}

For two words $u$ and $v$ and some $n\in\Nat$, we write $u\sim_n v$
when $u$ and $v$ have the same (scattered) subwords\footnote{Or
``subsequences'', not to be confused with ``factors''.} of length at
most $n$.  A language $L\subseteq A^*$ is \emph{piecewise-testable} if
it is closed under $\sim_n$ for some $n\in\Nat$.

Piecewise-testable (PT) languages were introduced more than forty
years ago in Imre Simon's doctoral thesis
(see~\cite{simon72,simon75,sakarovitch83}) and have played an important role
in the algebraic and logical theory of first-order definable
languages, see~\cite{pin86,DGK-ijfcs08,klima2011} and the references
therein.  They also constitute an important class of simple regular
languages with applications in learning theory~\cite{kontorovich2008},
databases~\cite{bojanczyk2012b}, linguistics~\cite{rogers2013},
etc. The concept of PT languages has been extended to
 richer notions of ``subwords''~\cite{zetzsche2018}, 
to
trees~\cite{bojanczyk2012b}, infinite words~\cite{perrin2004,carton2018b},
pictures~\cite{matz98}, 
or any combinatorial
well-quasi-order~\cite{goubault2016}.
\\

When a language $L\subseteq A^*$ is PT, we further say that it is
``$n$-PT'' if it is closed under~$\sim_n$, and the smallest such $n$
is called the PT \emph{height} of $L$, denoted
$h(L)$ in this article.

The height of piecewise-testable languages is a natural measure of
descriptive complexity.  Indeed, $h(L)$ coincides with the number of
variables needed in a $\calB\Sigma_1$ formula that defines
$L$~\cite{DGK-ijfcs08}.  In this article, the main question we address
is ``\emph{how can one bound the height of PT languages obtained by
natural language-theoretic operations?}'' Since the height of these
languages is a more robust measure than, say, their state complexity,
it can be used advantageously in the complexity analysis of problems
where PT languages are prominent.  As a matter of fact, our results
apply to, and were motivated by, open problems in the complexity
analysis of a logic of subwords, see~\Cref{sec-FO-basic}.

\subsubsection*{Related work.}
The height of PT languages has been used to measure the difference
between separable languages, see e.g.~\cite{hofman2015}.  Deciding
whether a DFA or a NFA $\calA$ recognises a $n$-PT language is
$\coNP$-complete or $\PSPACE$-complete respectively
(see~\cite{masopust2017} and the references therein).  The methods
underlying these algorithms usually provide a bound on $h(L)$ in
terms of $\calA$: Kl{\'\i}ma and Pol\'ak showed that $h(L)$ is
bounded by the maximal length of a simple path from an initial to a final state in $\calA$~\cite{klima2013}.  The currently best
bounds on $h(L)$ based on automata for $L$ have been obtained by
Masopust and Thomazo~\cite{masopust2015,masopust2016b}.

When $L$ is obtained by operations on other languages, very little is
known about PT heights. It is clear that $h(A^*\setminus
L)=h(L)$ and that $h(L\cup L')\leq\max(h(L),h(L'))$ but
beyond boolean operations, quotients, and inverse morphisms, there are
very few known ways of obtaining PT languages.

\subsubsection*{Our contribution.}
We provide upper and lower bounds on the PT height of finite languages and on
PT languages obtained by downward-closure (collecting all subwords of
all words from some $L$), upward-closure, and some related operations
(collecting words in $L$ that are minimal wrt the subword ordering,
etc.)  We also show that the incomparability relation preserves
piecewise-testability and we bound the PT heights of the resulting
languages.  Crucially, we show that these bounds are \emph{polynomial}
when expressed in terms of the PT height of the arguments.  One
important tool is a \emph{small-subword theorem} that shows how any
long word $u$ contains a short subword $u'$ that is
$\sim_n$-equivalent.  Reasoning about subwords involves {ad hoc}
techniques and leveraging the small-subword theorem to analyse
downward-closures or incomparability languages turns out to be
non-trivial. Subsequently, all the above results are used to prove
that $\FO^2(A^*,\subword)$, the two-variable logic of subwords, has
elementary complexity. For this logic, the decidability proof
in~\cite{KS-fosubw} did not come with an elementary complexity upper
bound because the usual measures of complexity for regular languages
can grow exponentially with each boolean combination of upward and
downward closures, and this is what prompted our investigation of PT
heights.

\subsubsection*{Outline of the article.}
\Cref{sec-basics} recalls the basic notions (subwords, Simon's
congruence, etc.)  and gives some first bounds relating PT
heights and minimal automata.  \Cref{sec-PT-levels} focuses on
finite languages and develops our main tool: the small-subword
theorem.  Sections~\ref{sec-PT-up} and~\ref{sec-PT-down} give bounds
for the height of PT languages obtained by upward and downward
closures, while \Cref{sec-PT-IL} considers the incomparability
relation and the resulting PT heights.  Finally, in
\Cref{sec-FO-basic} we apply these results to the complexity of
$\FO^2(A^*,\subword)$. In passing, we   characterise the expressive
power of the $\FO^2(A^*,\subword)$ logic, or equivalently of its
quantifier-free fragment, via  new notions of subword-recognizable and
piecewise-testable relations on words.




\section{Basic notions}
\label{sec-basics}

We consider finite words $u,v,...$ over a given finite alphabet $A$ of
letters like $a,b,\ldots$.  Concatenation of words is written
multiplicatively, with the empty word $\epsilon$ as unit.  We freely
use regular expressions like $(a b)^*+(b a)^*$ to denote regular
languages.

The length of a word $u$ is written $|u|$ while, for a letter $a\in
A$, $|u|_a$ denotes the number of occurrences of $a$ in $u$.  The set
of letters that occur in $u$ is denoted by $\alpha(u)$.  The set of
all words over $A$ is written $A^*$ and for $\ell\in\Nat$ we use
$A^{=\ell}$ and $A^{\leq \ell}$ to denote the subsets of all words of
length $\ell$ and of length at most $\ell$ respectively.

A word $v$ is a \emph{factor} of $u$ if there exist words $u_1$ and
$u_2$ such that $u = u_1 v u_2$.  If furthermore $u_1=\epsilon$ then
$v$ is a \emph{prefix} of $u$ and we write $v^{-1}u$ to denote the
residual $u_2$. If $u_2=\epsilon$ then $v$ is a \emph{suffix} and $u
\,v^{-1}$ is the residual.

\subsection{Subwords and superwords.}
We say that a word $u$ is a \emph{subword} (i.e., a subsequence) of
$v$, or equivalently that $v$ is a \emph{superword} of $u$, written
$u\subword v$, when $u$ is some $a_1\cdots a_n$ and $v$ can be written
as $v_0a_1 v_1 \cdots a_n v_n$ for some $v_0,v_1,\ldots,v_n\in A^*$,
e.g., $\epsilon\subword b b a\subword a b a b a$.

We write $u\ssubword v$ for the associated strict ordering, where
$u\neq v$. Two words $u$ and $v$ are \emph{incomparable} (with respect
to the subword relation), denoted $u \perp v$, if $u \not \subword v$
and $v \not \subword u$.  Factors are a special case of subwords.

With any $u\in A^*$ we associate the upward and downward closures,
$\LL{u}$ and $\DD{u}$, given by
\begin{xalignat*}{2}
\LL{u}&\egdef\{v\in A^* ~|~ u\subword v\}    \:,
&
\DD{u}&\egdef\{v\in A^* ~|~ v\subword u\}    \:.
\end{xalignat*}
(Formally, one should write $\LLA{u}$ since the definition depends on
the alphabet at hand, but we will leave $A$ implicit: it will always
be clear from the context.)  For example, $\DD{a b}=\{a b, a, b,
\epsilon\}$ and $\LL{a b}=A^*aA^*bA^*$.

This is generalised to the closures of whole languages, via $\up
L=\bigcup_{u\in L}\LL{u}$ and $\down L=\bigcup_{u\in L}\down u$.  The
Kuratowski closure axioms are satisfied:
\[
\up \emptyset=\emptyset
\:,
\quad
L\subseteq\up L=\up\up L
\:,
\quad
\up\bigl(\bigcup_i L_i\bigr)=\bigcup_i\up L_i
\:,
\quad
\up \bigl(\bigcap_i \up L_i\bigr)=\bigcap_i\up L_i
\:,
\]
and similarly for downward closures.
We say that a language $L$ is \emph{upward-closed} if $L = \up L$, and
\emph{downward-closed} if $L = \down L$.  Note that a language is
upward-closed if, and only if, its complement is downward-closed.

A variant of the closure operations is based on the strict ordering:
we let
\begin{xalignat*}{4}
\up_< u & \egdef\{v~|~u\ssubword v\}
\:,
&
\up_< L & \egdef\bigcup_{u\in L}\up_< u
\:,
&
\down_< u & \egdef\{v~|~v\ssubword u\}
\:,
&
\down_< L & \egdef\bigcup_{u\in L}\down_< u
\:.
\end{xalignat*}
While these are not closure operations, the languages $\up_<L$ and
$\down_< L$ are upward-closed and downward-closed, respectively.
Since upward-closed and downward-closed languages are regular (Haines
Theorem~\cite{haines69}, also a corollary of Higman's
Lemma~\cite{higman52}) we conclude that $\up L$, $\down L$, $\up_<L$
and $\down_<L$ are regular for any $L$.

Finally we further define
\begin{xalignat*}{1}
  I(L) &\egdef \{ u\in A^* ~|~ \exists v \in L: u\perp v\}      \:.
\end{xalignat*}
Thus $I(L)$ collects all words that are incomparable with \emph{some
word} in $L$. For example, $I(a^*b^*)=A^+$ and $I((aaa)^+) =
A^*\setminus a^*$.

\subsection{Recognizable and rational relations over words.}
Recall that a binary relation $R\subseteq A^*\times A^*$ is
\emph{rational} if it can be defined via an (asynchronous
nondeterministic) transducer or, equivalently, via a regular
expression using elements of $A^*\times A^*$, unions, concatenations
and Kleene stars, see, e.g.~\cite[Chap.\ 3]{berstel79} or
\cite[Chap.\ 4]{sakarovitch2009}.  It is well-known that, while
$\Rat(A^*\times A^*)$ is (effectively) closed under composition ---as
well as union, concatenation, and Kleene star,--- it is not closed
under complement or intersection.

For example, we can define equality over $A^*$ as well as the subword
relations (strict and non strict) via the following regular
expressions:\footnote{The expression for $\ssubA$ uses the
concatenation, denoted $R\cdot R'$,
of relations, not their composition $R\circ R'$.}
\begin{xalignat}{3}
\label{eq-subword-rational}
\idA               &= \left(\bigcup_{a\in A}\pair{a}{a}\right)^*
\!\!,
&
\subA        &= \left(\bigcup_{a\in A}\pair{\epsilon}{a}\cup\pair{a}{a}\right)^*
\!\!,
&
\ssubA &= \subA\cdot\left(\bigcup_{a\in A}\pair{\epsilon}{a}\right)\cdot\subA
.
\end{xalignat}
Since $\subA$ and $\ssubA$ can even be defined by \emph{deterministic}
transducers, we deduce that their complements, $\nsubA$ and $\nssubA$, are
rational. Finally, let us mention the following result:
\begin{propC}[\cite{KS-fosubw}]
\label{perp-rational}
The incomparability relation $\perpA \subseteq A^*\times A^*$ is
rational.
Consequently
 $I(L)$, i.e.,
the image of $L$ by $\perpA$, is effectively regular for any regular
$L$.
\end{propC}
We note that proving Proposition~\ref{perp-rational} cannot rely
on the  characterisation $\perpA = \nsubA\cap\nsupA$
since the intersection of two rational relations is in general not
rational, even when the two relations are given by deterministic
transducers.
\\

The rational relations over $A^*$ encompass the special case of the
\emph{recognizable relations}. Recall that $R\subseteq A^*\times A^*$
is recognizable ---in the standard way, i.e., ``{by some morphism to a
finite monoid}''--- if it is a finite union $R=L_1\times L'_1\cup
\cdots \cup L_m\times L'_m$ of cartesian products where all $L_i$'s
and $L'_i$'s are regular languages over $A$.  Recognizable relations
are rational but the converse does not hold, for example, the equality
relation $\idA$ is not recognizable. We shall use the well-known and
easy-to-see fact that $\Rec(A^*\times A^*)$ is (effectively) closed
under boolean operations.

\subsection{Simon's congruence.}
For $n \in\Nat$ and $u,v \in A^*$, we let
\begin{xalignat}{1}
\label{eq-def-sim-n}
   u \sim_n v &\equivdef \DD{u}\cap A^{\leq n}=\DD{v}\cap A^{\leq n}
\:.
\end{xalignat}
In other words, $u\sim_n v$ if $u$ and $v$ have the same subwords of
length at most $n$.  For example $abab\sim_1 aabb$ (both words use the
same letters) but $abab\not\sim_2 aabb$ ($ba$ is a subword of $abab$,
not of $aabb$).  Note that $u\sim_0 v$ for any $u,v$, and $u\sim_n u$
for any $n$.
We write $[u]_n$ for the equivalence class of $u\in A^*$ under
$\sim_n$.  Note that each $\sim_n$, for $n=1,2,\ldots$, has finite
index~\cite{simon75,sakarovitch83}.

We further let
\begin{xalignat}{1}
\label{eq-def-simstar}
   u \simonRstar v & \equivdef u\sim_n v \land u\subword v
\:.
\end{xalignat}
Note that $\simonRstar$ is stronger than $\sim_n$. Both relations are
(pre)congruences: $u\sim_n v$ and $u'\sim_n v'$ imply $uu'\sim_n vv'$,
while $u\simonRstar v$ and $u'\simonRstar v'$ imply $uu'\simonRstar
vv'$.
The equivalence $\sim_n$, introduced in~\cite{simon72}, is called Simon's congruence of order $n$.

The following properties will be useful:
\begin{lemma}
\label{lem-useful}
For all $u,v,v' \in A^*$ and $a,b \in A$:
\begin{enumerate}

\item
\label{it-convex}
If  $u \simonRstar v$ then $u\sim_n w$ for
all $w\in A^*$ such that $u\subword w\subword v$;

\item
\label{it-carac-richn}
When $n>0$, $u\sim_n uv$ if, and only if, there exists a
factorization $u=u_1 u_2\cdots u_n$ such that $\alpha(u_1)
\supseteq \alpha(u_2) \supseteq \cdots \supseteq \alpha(u_n) \supseteq
\alpha(v)$;

\item
\label{it-diff-let}
If $uav\sim_n ubv'$ and $a\neq b$ then
$ubav\sim_n ubv'$ or $uabv'\sim_n uav$ (or both);

\item
\label{it-upperbound}
If  $u \sim_n v$ then there exists $w \in
A^*$ such that $u \simonRstar w$ and $v \simonRstar w$;

\item
\label{it-shorter}
If $u\sim_n v$ and $|u|<|v|$ then there exists some $v'$ with
$|v'|=|u|$ and such that $u\sim_n v'\subword v$;

\item
\label{it-pumping}
If $u v \sim_n uav$ then
$uv \sim_n u a^\ell v$ for all $\ell\in\Nat$;

\item
\label{it-sing-inf}
Every equivalence class of $\sim_n$ is a singleton or is infinite.

\end{enumerate}
\end{lemma}
\begin{proof}
(\ref{it-convex}) is by combining Eq.~\eqref{eq-def-sim-n} with
$\DD{u}\subseteq \DD{w}\subseteq \DD{v}$;
(\ref{it-carac-richn}--\ref{it-upperbound}) are Lemmas 3, 5, and 6 from~\cite{simon75};
(\ref{it-shorter}) is an immediate consequence of Theorem~4 from
\cite[p.~91]{simon72}, showing that all minimal (wrt.\ $\subword$) words in $[u]_n$ have
the same length ---see also Theorem~6.2.9 from
\cite{sakarovitch83};
(\ref{it-pumping}) is in the proof of Corollary~2.8 from~\cite{sakarovitch83};
(\ref{it-sing-inf}) follows from (\ref{it-convex}), (\ref{it-upperbound}) and (\ref{it-pumping}).
\end{proof}

\subsection{Piecewise-testable languages.}
A language $L\subseteq A^*$ is piecewise-testable (or PT) if it if
closed under $\sim_n$ for some $n$ (and then we say that it is
$n$-piecewise-testable, or $n$-PT).  Note that if $L$ is $n$-PT, it is
also $m$-PT for any $m>n$. We write $h(L)$ for the smallest $n$
---called the \emph{height} of $L$--- such that $L$ is $n$-PT, letting
$h(L)=\infty$ when $L$ is not PT.  Finally, we write $\PT$ for the
class of piecewise-testable languages (over some alphabet $A$) and
$\PT_n$ for the class of languages with height at most $n$, so that
$\PT_0\subseteq \PT_1\subseteq \cdots\PT_n\subseteq \cdots\PT$ form a
 hierarchy of varieties of regular languages.

\begin{fact}[Alternative characterisations of PT and $n$-PT languages]
Let $L\subseteq A^*$. The following are equivalent:
\label{PT-characterizations}
\begin{enumerate}
\item
$L$ is $n$-PT (i.e., closed under $\sim_n$);
\item
$L$ is a finite union $[u_1]_n\cup[u_2]_n\cup \cdots\cup [u_m]_n$
of $\sim_n$ classes;
\item
$L$ is a finite boolean combination of principal filters
  $A^*a_1A^*a_2A^*\cdots a_\ell A^*$ (i.e., of closures $\up
  a_1a_2\cdots a_\ell$) with $\ell\leq n$;
\item
$L$ is definable in the $\calB\Sigma_1[<]$ fragment\footnote{That is,
the boolean closure of the existential fragment.} of first-order logic
  over words, via a formula involving only $n$
  variables~\cite{DGK-ijfcs08}.
\end{enumerate}
The following are equivalent:
\begin{enumerate}
\setcounter{enumi}{4}
\item
$L$ is PT;
\item
$L$ is recognised by a finite and $\calJ$-trivial
  monoid~\cite{simon75,straubing88,henckell2000,klima2011};
\item
$L$ is regular and its minimal DFA is partially ordered and satisfies the
  UMS property~\cite{simon75,trahtman2001};
\item
$L$ is regular and its minimal DFA is acyclic and locally
confluent~\cite{klima2013}.
\end{enumerate}
\end{fact}
The characterisations (3), (4), (7) and (8) are useful for showing
that a language is PT ---or even $n$-PT in the case of (3) and (4)---.  For
example, with alphabet $A=\{a,b,c\}$, the language $a^+b^*$ can be
defined via required and excluded minors, as in:
\begin{gather}
\label{ex-PT-by-minors}
u\in a^+b^* \iff
a\subword u\land b a\not\subword u\land c\not\subword u
\:.
\end{gather}
This definition of $a^+b^*$ directly translates into a $\calB\Sigma_1$
formula, or into a finite boolean combination of filters:
\begin{gather}
\label{ex-PT-by-filters}
a^+b^* = \up a \setminus\up b a \setminus \up c
= A^* a A^* \setminus A^* b A^* a A^* \setminus A^* c A^*
\:.
\end{gather}
Furthermore, since the minors in Eq.~\eqref{ex-PT-by-minors} have
length at most 2 (they are $a$, $ba$, and $c$), we conclude that
$h(a^+b^*)\leq 2$.

The characterisations (7) and (8) are also useful for showing that a
language is not PT (see examples in \Cref{ssec-closure-props} below). Finally (1) is very useful
for showing that a language is not $n$-PT for a given $n$: by
exhibiting two words $u\sim_n v$ such that $u\in L$ and $v\not\in L$,
one proves that $L$ is not saturated by $\sim_n$.  E.g., one sees that
$a^+b^*$ is not 1-PT since $ab\sim_1 ba$ while only $ab$ is in
$a^+b^*$. We may now conclude that $h(a^+b^*)=2$.
\\

Some examples of (families of) PT languages are:
\begin{description}
\item[All finite languages] $u\sim_{n}v$ and $n>|u|$ imply $u=v$.
  Thus $[u]_n=\{u\}$ and any $L=\{u_1,\ldots,u_m\}\subseteq
  A^{<n}$ can be expressed as $L=[u_1]_n\cup \cdots\cup[u_m]_n$.
  By characterisation (2), $L$ is PT.
\item[All upward-closed languages] By Haines
  Theorem~\cite{haines69,higman52}, any language $L\subseteq A^*$ has
  finitely many minimal elements (wrt $\subword$), i.e., $\min(L)$ is
  some $\{u_1,\ldots,u_m\}$. This entails $\up L=\up u_1\cup \cdots
  \cup\up u_m$, which is PT by characterisation (3).
\item[All downward-closed languages] They are the complements of
  upward-closed languages, hence PT again by characterisation (3).
\end{description}
In sections \ref{sec-PT-levels}, \ref{sec-PT-up} and \ref{sec-PT-down}
respectively, we analyse the PT heights of languages belonging to the
above three families.

\subsection{Piecewise-testable relations}

Following the generic pattern laid out in~\cite{goubault2016}, we say
that a relation $R\subseteq A^*\times A^*$ is
\emph{piecewise-testable} if it is a finite boolean combination of principal
filters $\up(u,v)$ in the product ordering $(A^*\times
A^*,\subword\times \subword)$.  The relation is $n$-PT if the boolean
combination only uses filters $\up(u,v)$ with $|(u,v)|\leq n$, where
we define $|(u,v)|\egdef\max(|u|,|v|)$.

Since $\up(u,v)=(\up u)\times(\up v)$, we see that piecewise-testable
relations are recognizable. Using
\begin{align*}
(A^*\times A^*)\setminus\up(u,v)
&\;=\;
(A^*\setminus \up u)\times A^*
\:\cup\: A^*\times (A^*\setminus \up v)
\\
&\;=\;
(\neg \up u)\times \up\epsilon
\:\cup\: \up\epsilon\times (\neg \up v)
\end{align*}
and de~Morgan's laws, we further see that a finite boolean combination
of filters $\up(u,v)$ can be written as some $\bigcup_i\bigcap_j (\pm
\up u_{i,j})\times (\pm \up v_{i,j})$, i.e., $\bigcup_i\bigl((\bigcap_j \pm
\up u_{i,j})\times (\bigcap_j \pm \up v_{i,j})\bigr)$.  Finally, any $n$-PT
relation can be written under the form $R=L_1\times L'_1\cup
\cdots\cup L_m\times L'_m$ where all $L_i$'s and $L'_i$'s are $n$-PT
languages. Hence PT relations form a subclass, denoted $\PT(A^*\times
A^*)$ of $\Rec(A^*\times A^*)$. We shall not use PT relations until
\Cref{sec-FO-basic} and, for the moment, keep focused on PT languages.

\subsection{Closure properties of PT languages.}
\label{ssec-closure-props}

By definition (see \Cref{PT-characterizations}), any class $\PT_n$ is
closed under boolean operations.  Furthermore, $\PT_n$ is also closed
under (left and right) quotients and under inverse
morphisms~\cite{therien81}.  In terms of PT height, the
above statements can be written as
\begin{xalignat}{3}
\label{eq-basic-1}
h\Bigl(\bigcup_i L_i\Bigr)&\leq \max\: \{h(L_i)\}_i\:,
&
h\Bigl(\bigcap_i L_i\Bigr)&\leq \max\: \{h(L_i)\}_i\:,
&
h(\neg L)&=h(L)\:,
\\
\label{eq-basic-22}
h(\rho^{-1}(L))&\leq h(L) \makebox[0cm][l]{ for $\rho:A^*\to B^*$ a morphism,}
&&
&
h(u^{-1}L\,v^{-1})&\leq h(L)\:.
\end{xalignat}
Note that we can allow arbitrary unions and intersections in
Eq.~\eqref{eq-basic-1} since, for fixed $A$, there are only finitely many
languages in $\PT_n$.

Let us also mention that $\PT_n$ is closed under taking mirror images:
writing $\overleftarrow{a_1a_2\cdots a_\ell}$ for the mirror word
$a_\ell\cdots a_2a_1$, and letting
$\overleftarrow{L}=\{\overleftarrow{u}~|~u\in L\}$ for a language
$L\subseteq A^*$, one has $h\bigl(\overleftarrow{L}\bigr)=h(L)$.

Beyond that, $\PT$ is not closed under any of the usual
language-theoretic operations as we now illustrate.
\begin{description}

\item[Concatenation and prefixing] $a(a+b)^*$ is not PT: it is not
  closed under any $\sim_k$ since it contains
  $(ab)^k$ but not $b(ab)^k$ while $(ab)^k\sim_k b(ab)^k$.\footnote{An
  alternative proof is by observing that the minimal automaton for
  $(a+b)^*a$ ---the mirror language--- is not acyclic.}
Since $(a+b)^*$ is PT, we see that the class $\PT$ is not closed under concatenation, even
  in the special case of prefixing with $a$.

\item[Kleene star] $\PT$ is not closed under Kleene star (recall that
  $\PT$ is a subvariety of the star-free languages): $aa$ is finite
  hence PT but $(aa)^*$ is not PT since its minimal automaton is not acyclic.

\item[Shuffle product] $ab^*$ and $a^*$ are PT but their shuffle
  product $ab^*\shuffle a^*=a(a+b)^*$ is not as just shown, see~\cite{HS-ipl2019} for PT shufflings.

\item[Conjugacy] Recall that the conjugates of $u$ are
  $\widetilde{u}\egdef\{u_2u_1~|~u=u_1u_2\}$, and we extend with
  $\widetilde{L}=\bigcup_{u\in L}\widetilde{u}$.  Now $L=ac(a+b)^*$ is
  PT but $\widetilde{L} =(a+b)^*ac(a+b)^*+c(a+b)^*a$ is not.

\item[Renaming] $c(a+b)^*$ is PT but applying the renaming $c\mapsto
  a$ yield $a(a+b)^*$ which is not.

\item[Erasing one letter] This operation can be seen as the inverse
  of $L\mapsto L\shuffle A$ where an arbitrary letter is inserted at
  an arbitrary position. Now $ac(a+b)^*$ is PT but erasing one letter
  yields $(a+c+ac)(a+b)^*$ which is not PT.
\end{description}

Finally, we are only aware of one more positive instance of a closure
property, ``$I(L)$ is PT when $L$ is'', but proving this is the topic
of \Cref{sec-PT-IL}.

\subsection{Relating PT height and state complexity.}
For regular languages, a standard measure of descriptive complexity is
\emph{state complexity}, denoted $\stc(L)$, and defined as the number of
states of the minimal DFA for $L$~\cite{yu2005}.

The bounds given in Eqs.~\eqref{eq-basic-1} and \eqref{eq-basic-22}
let us contrast the height of a PT language with its state complexity.
If $L$ is a PT language, one has $h(L)\leq\stc(L)$
(equality occurs e.g. when $L=\{a^\ell\}$) since $\stc(L)$ bounds the
depth of the minimal automaton for $L$, i.e., the maximum length of a simple path from
the initial to some final state, which in turns bounds
$h(L)$~\cite{klima2013,masopust2015}.

In the other direction, we can prove
\begin{theorem}
\label{thm-size-DFA-for-nPT}
Let $A$ be an alphabet of size $k$ with $k > 1$. Suppose $L \subseteq
A^*$ is $n$-PT. Then the minimal DFA for $L$ has at most $m$
states,\footnote{It is shown in~\cite{masopust2015,masopust2017}
 that the depth (not the size) of the minimal DFA is bounded by
$\binom{n+k}{n}-1$.} where
\[
\log m = k \left(\frac{n+2k-3}{k-1} \right)^{k-1} \log n \log k
\:.
\]
Here $\log$ means $\log$ to the base $2$. Thus, for fixed $k$, $\stc(L)$ is   $2^{O(n^{k-1} \log n)}$, where $n=h(L)$.
\end{theorem}
\begin{proof}
Write $L$ under the form $L=[u_1]_n\cup \cdots \cup[u_m]_n$ and
consider the $\sim_n$-\emph{canonical DFA}, i.e., the DFA
$\calA=(Q,A,i,.,F)$ whose states are all classes $[w]_n$ for $w\in
A^*$ and with transitions given by $[w]_n.a=[wa]_n$. With initial
state $i=[\epsilon]_n$, this automaton reaches $[w]_n$ upon reading
$w$. With accepting states $F=\{[u_1]_n,\ldots,[u_m]_n\}$, it
recognises exactly $L$. In \cite{KKS-ipl} we showed that the number of
equivalence classes of $\sim_n$, i.e., $|Q|$, is bounded by $m$. This
in turns bound $\stc(L)$.
\end{proof}

The general situation is that $h(L)$ can be much smaller than
$\stc(L)$ as we shall see in the following sections.  More importantly,
PT height is a more robust measure than state complexity.  For
example, state complexity can increase exponentially when building
boolean combinations of regular languages (while PT height does not
increase). This classic phenomenon occurs even when we restrict to PT
languages, and even if we use nondeterministic state complexity, see
the examples in~\cite{KNS-tcs2016}. This is not limited to boolean
operations: for example, we saw that
$h\bigl(\overleftarrow{L}\bigr)=h(L)$ but
$\stc\bigl(\overleftarrow{L}\bigr)$ cannot be bounded by a polynomial
of $\stc(L)$, even in the case of finite, hence PT,
languages~\cite{salomaa2004}.




\section{PT height of words and the small-subword theorem}
\label{sec-PT-levels}

Our starting point is an analysis of the PT height of single words.
It is clear that any singleton language $\{u\}$ is PT since $\{u\} =
\up u\setminus \bigcup_{v\in \{u\}\shuffle A} \up v$, which entails
$h(\{u\}) \leq |u|+1$.  Here we used a  \emph{shuffle
product} notation, 
$\{u\}\shuffle A$, to denote $\{v~:~u\subword v\land
|v|=|u|+1\}$, i.e., the set of all words obtained from $u$ by
inserting, at some position, one letter from $A$.
Below we often omit set-theoretical parentheses when
denoting singletons, writing e.g.\ ``$h(u)$'' or ``$u\shuffle A$''.
\\

The $|u|+1$ upper bound for $h(u)$ is tight. For example,
\begin{gather}
\label{eq-h-aell}
h(a^\ell) = \ell+1
\:.
\end{gather}
(To see
that $h(a^\ell)>\ell$, one notes that
$a^\ell\sim_\ell a^{\ell+1}$.)
However, words on more than one letter can generally be described within some
PT height lower than their length. For example
\[
\{aabb\}= (\up aa\cap\up bb)\setminus ( \up ba \cup \up aaa \cup \up
bbb)
\:,
\]
showing $h(aabb)\leq 3$. (Note that $h(aabb)>2$ since $aabbb\sim_2
aabb$, thus  $h(aabb)=3$.)

It turns out that the PT height of words can be much lower than their
length as we shall see in \cref{ssec-uk}, but before considering lower
bounds on $h(u)$, let us make some easier observations.

\begin{proposition}
\label{hu-in-ptime}
The PT-height of a word can be computed in polynomial time.
\end{proposition}
\begin{proof}
Following~\cite{sakarovitch83}, we let
$\delta(v,v')\egdef\max\{n~|~v\sim_n v'\}$ for any two words $v,v'\in A^*$.
We now claim that
\begin{gather}
\label{eq-h-delta}
          h(u) = 1+\max \{\delta(u,v)~|~{v\in u\shuffle A}\}\:.
\end{gather}
To prove the claim, we note that $u\neq v$ entails $\delta(u,v)<h(u)$
since, by definition of $h(u)$, $u\sim_{h(u)}v$ entails $u=v$.
In the other direction, let
$n=h(u)-1$ so that $[u]_n$ is not a singleton.  Then $[u]_n$ is
infinite (Lemma~\ref{lem-useful}~(\ref{it-sing-inf})) and in particular
contains some word $w$ with $u\ssubword w$
(Lemma~\ref{lem-useful}~(\ref{it-upperbound})). In fact $[u]_n$ contains
all words between $u$ and $w$ (Lemma~\ref{lem-useful}~(\ref{it-convex}))
hence some $w'\in u\shuffle A$.  Thus the right-hand side of
\eqref{eq-h-delta} is at least $1+\delta(u,w')$, i.e., at least
$n+1=h(u)$.

We have thus reduced the computation of $h(u)$ to polynomially many
$\delta(u,v)$ computations and now rely on the fact that $\delta$ can
be computed in polynomial time~\cite{simon2003}.\footnote{It is also
possible to directly compute $h(u)$ in time and space $O(|h|\cdot
|A|)$ by adapting the techniques used in~\cite{fleischer2018}, where
the goal is to compute a canonical $\sim_n$-equivalent for $u$.}
\end{proof}

\Cref{hu-in-ptime}
can be used to compute the PT height of finite languages in polynomial
time: for such
languages, the inequality in Eq.~\eqref{eq-basic-1} becomes
\begin{gather}
\label{eq-h-finite}
h(\{u_1,\ldots,u_m\}) = \max \{h(u_1),\ldots,h(u_m)\}
\:.
\end{gather}
Indeed, $h(\{u_1,\ldots,u_m\})=n$ implies $[u_i]_n \subseteq
\{u_1,\ldots,u_m\}$ for any $i$.  Thus $[u_i]_n$ is a singleton in
view of Lemma~\ref{lem-useful}~(\ref{it-sing-inf}).  Hence $[u_i]_n=\{u_i\}$ and $h(u_i)\leq
n$.

\subsection{Words with low PT height}
\label{ssec-uk}

We introduce a family of words with ``low PT height'' that will
be used repeatedly in later sections.
Let $A_k=\{a_1,\ldots,a_k\}$ be a $k$-letter alphabet.  We define a
word $U_k\in A_k^*$ by induction on $k$ and parameterized by a
parameter $\lambda\in\Nat$.  We let $U_0\egdef\epsilon$ and, for $k>0$,
$U_k\egdef (U_{k-1}a_k)^\lambda U_{k-1}$. For example, for $\lambda=3$ and
$k=2$, one has $U_2=a_1^3a_2a_1^3a_2a_1^3a_2a_1^3=
a_1a_1a_1a_2a_1a_1a_1a_2a_1a_1a_1a_2a_1a_1a_1$.
\\

The rest of \Cref{ssec-uk} establishes the following bounds for any
$k,\lambda\in \Nat$:
\begin{xalignat}{3}
\label{eq-bounds-uk}
|U_k|&=(\lambda+1)^k-1\:,
&
h(U_k)&=k\lambda+1\:,
&
h(\down U_k) &= \lambda(\lambda+1)^{k-1} +1 \:.
\end{xalignat}
The first equality, $|U_k|=(\lambda+1)^k-1$, is easily seen by induction
on $k$.

To show $h(U_k)=k\lambda+1$ we use some auxiliary languages
$P_k,N_k\subseteq A_k^*$ defined inductively by the following
expressions:
\begin{xalignat}{2}
P_0&\egdef\{\epsilon\},
&
N_0&\egdef\emptyset,
\\
\shortintertext{and, for $k>0$,}
P_k &\egdef \sum_{i=0}^\lambda a_k^i \cdot P_{k-1} \cdot a_k^{\lambda-i}
,
&
N_k &\egdef a_k^{\lambda+1} + \sum_{i=0}^\lambda a_k^i \cdot N_{k-1} \cdot a_k^{\lambda-i}.
\end{xalignat}
The words in $P_k$ and $N_k$ are used as positive, and
respectively negative, constraints in the following claim.
\begin{claim}
\label{lem-uk-pk-nk}
For any $k\in\Nat$ and $u\in A_k^*$:
\begin{gather}
\label{eq-uk-pk-nk}
\bigl(\bigwedge_{v\in P_k}v\subword u\bigr)\:\land\: \bigl(\bigwedge_{w\in
  N_k} w\not\subword u\bigr)
\iff
u=U_k.
\end{gather}
\end{claim}
\begin{proof}
By induction on $k$. For $k=0$, $A_0$ is empty and there is only one
word in $A_0^*$, namely $u=U_0=\epsilon$. It satisfies the positive
constraint $\epsilon\subword u$ (from $P_0$) and there are no negative constraints in $N_0$.

Assume now that $k>0$ and that the claim holds for $k-1$. We prove the
left-to-right implication: Since $P_k$ is not empty, the $P_k$
constraints $a_k^i va_k^{\lambda-i}\subword u$ imply that $|u|_{a_k}\geq \lambda$.
However the $N_k$ constraint $a_k^{\lambda+1}\not\subword u$ implies that
$u$ contains exactly $\lambda$ occurrences of $a_k$ and can be written
$u=v_0 a_k v_1 a_k \cdots a_k v_\lambda$ with $v_i\in A_{k-1}^*$ for all
$i=0,\ldots,\lambda$.

Consider some fixed $v_i$: for any $v\in P_{k-1}$ it holds that
$v\subword v_i$ since $a_k^iv a_k^{\lambda-i}\subword u$. Similarly
$w\not\subword v_i$ for any $w\in N_{k-1}$ since $a_k^i w
a_k^{\lambda-i}\not\subword u$. The induction hypothesis now yields $v_i=U_{k-1}$,
thus $u=U_{k-1}a_kU_{k-1}\cdots a_kU_{k-1}=U_{k}$. The right-to-left
implication should now be clear and can be left to the reader.
\end{proof}
\Cref{lem-uk-pk-nk} entails $h(U_k)\leq k\lambda+1$
since the words in $P_k$ have length $k\lambda$ and the words in $N_k$ have
length at most $k\lambda+1$.

It remains to show that $h(U_k)> k\lambda$, i.e., that $\{U_k\}$ is
not closed under $\sim_{k\lambda}$. For this
we factor $U_k$ under the form
\begin{gather}
U_k=(U_{k-1}a_k)^\lambda(U_{k-2}a_{k-1})^\lambda(U_{k-3}a_{k-2})^\lambda\ldots (U_0a_1)^\lambda
\:.
\end{gather}
Using Lemma~\ref{lem-useful}~(\ref{it-carac-richn}), this factorization in $k\lambda$ factors
involving decreasing alphabets proves $U_k\sim_{k\lambda} U_ka_1$ and
concludes the proof of the second equality in Eq.~\eqref{eq-bounds-uk}.
\\

To prove the third equality in Eq.~\eqref{eq-bounds-uk}, we write
$L_k$ for $|U_k|_{a_1}$ and note
that $L_k=(\lambda+1)L_{k-1}$ when $k>1$.
\begin{claim}
\label{u-subw-u-k-r}
For any $k,r\in\Nat$ and $u\in A_k^*$, if $u\subword U_k^r$, then $h(u)\leq 1+r L_k$.
\end{claim}
\begin{proof}
By induction on $k$. For $k\leq 1$, $u\subword U_k^r=a_1^{rL_k}$
requires $u=a_1^\ell$ with $\ell\leq rL_k$. Eq.~\eqref{eq-h-aell} then
gives $h(u)=1+\ell\leq 1+rL_k$.

So assume $k>1$.  Let $m=|u|_{a_k}$ and factor $u$ as
$u_0a_k u_1 a_k \ldots a_k u_m$ so that $u_i\in A_{k-1}^*$ for all
$i=0,\ldots,m$.
We then derive a PT-characterisation of $u$ from
PT-characterisations of the $u_i$'s: $u$ is the only word
in $A^*$ that satisfies
\begin{gather}
a_k^m\subword u \land
a_k^{m+1}\not\subword u \land
\bigwedge_{i=0}^m\:\bigwedge_{w\in A_{k-1}^{\leq h(u_i)}}
(a_k^i w a_k^{m-i}\subword u\iff w\subword u_i)
\:.
\end{gather}
We deduce that $h(u)\leq \max(m+1,m+h(u_0),\ldots,m+h(u_m))=
m+\max(1,h(u_0),\ldots,h(u_m))$.

Now recall that $u$ has $m$ occurrences of $a_k$ while $U_k^r$ has
$r(\lambda+1)$. This implies that any  $u_i$ in the decomposition of $u$
is a subword of $U_{k-1}^{r'}$ for
$r'=r(\lambda+1)-m$, so, by induction hypothesis,
$h(u_i)\leq 1+r' L_{k-1}$.
Assuming $k>1$, we thus have
\begin{align*}
h(u)\leq m+1+r' L_{k-1}
 &= m+1+[r(\lambda+1)-m]L_{k-1}
\\
& = 1 + m[1-L_{k-1}] + r(\lambda+1)L_{k-1}
 \leq 1 + r L_k
\:,
\end{align*}
 establishing the claim.
\end{proof}
\begin{corollary}
$h(\down U_k^r)=1+r L_k$, and thus in particular, $h(\down U_k) = 1 + L_k$.
\end{corollary}
\begin{proof}
Since $\down U_k^r$ is finite, \Cref{u-subw-u-k-r} and
Eq.~\eqref{eq-h-finite} entail $h(\down U_k^r)\leq 1+r L_k$.  On the
other hand, $a_1^{r L_k}\in \down U_k^r$. Hence $h(\down U_k^r)\geq
h(a_1^{r L_k})=1+r L_k$.
\end{proof}

\subsection{Rich words and rich factorizations}
Assume a fixed $k$-letter alphabet $A$.
We say that a word $u$ is \emph{rich} if $\alpha(u)=A$, i.e.,  the $k$
letters of $A$
all occur in $u$, and that it is \emph{poor} otherwise. For $\ell\in\Nat$, we
further say that $u$ is $\ell$-rich if it can be written as a
concatenation $u=r_1\cdots r_\ell u'$ where the $\ell$ factors
$r_1,\ldots,r_\ell$ are rich.

The \emph{richness} of $u$ is the largest
$\ell\in\Nat$ such that $u$ is $\ell$-rich. Note that having
$|u|_a\geq \ell$ for all letters $a\in A$ does not imply that $u$
is $\ell$-rich.
\begin{lemma}
\label{lem-rich-and-sim}
If $u_1$ and $u_2$ are respectively $\ell_1$-rich and $\ell_2$-rich,
then $v\sim_n v'$ implies $u_1 v u_2\sim_{\ell_1+n+\ell_2}u_1 v' u_2$.
\end{lemma}
\begin{proof}
A subword $x$ of $u_1 v u_2$ can be decomposed as $x=x_1 y x_2$ where
$x_1$ is the longest prefix of $x$ that is a subword of $u_1$ and $x_2$
is the longest suffix of the remaining $x_1^{-1} x$ that is a subword
of $u_2$. Thus $y\subword v$ since $x\subword u_1 v u_2$. Now, since
$u_1$ is $\ell_1$-rich, we have $|x_1|\geq \min(\ell_1,|x|)$, and similarly $|x_2|\geq \min(\ell_2,|x_1^{-1}x|)$. Finally
$|y|\leq n$ when $|x|\leq \ell_1+n+\ell_2$, and then
$y\subword v'$ since $v\sim_n v'$, entailing $x\subword u_1 v' u_2$. A
symmetrical reasoning shows that subwords of $u_1 v' u_2$ of length $\leq
\ell_1+n+\ell_2$ are subwords of $u_1 v u_2$ and we are done.
\end{proof}

The \emph{rich factorization} of $u\in A^*$ is the decomposition
$u=u_1 a_1 \cdots u_m a_m v$
defined by induction in the following way: if $u$ is poor, we
let $m=0$ and $v=u$; otherwise $u$ is rich, we let $u_1 a_1$ (with
$a_1\in A$) be the
shortest  prefix of $u$ that is rich and let
$u_2 a_2\cdots u_m a_m
v$ be the rich factorization of the remaining suffix $(u_1a_1)^{-1}u$.
By construction $m$ is the richness of $u$.
E.g.,
assuming $k=3$ and $A=\{a,b,c\}$, the following is a rich factorization with $m=2$:
\[
\obracew{bbaaabbccccaabbbaa}{u}
=\obracew{bbaaabb}{u_1}\cdot c\cdot
\obracew{cccaa}{u_2}\cdot b \cdot\obracew{bbaa}{v}
\]
Note that, by construction, $u_1,\ldots,u_m$ and $v$ are poor.

\begin{lemma}
\label{lem-rich-middle}
Consider two words $u,u'$ of richness $m$ and with rich factorizations $u=u_1 a_1 \cdots u_m a_m v$ and $u'=u'_1 a_1
\cdots u'_m a_m v'$. Suppose that  $v \sim_n v'$ and that
$u_i \sim_{n+1} u'_i$  for all $i=1,\ldots,m$. Then $u \sim_{n+m} u'$.
\end{lemma}
\begin{proof}
Since each factor $u_i a_i$ is rich, one gets
\begin{align*}
u_1 a_1 u_2 a_2\cdots u_m a_m v
&\;\sim_{n+m}\;  u'_1 a_1 u_2 a_2\cdots u_m a_m v
\;\sim_{n+m}\;  u'_1 a_1 u'_2 a_2\cdots u_m a_m v
\\
\;\sim_{n+m}\;\cdots
&\;\sim_{n+m}\;  u'_1 a_1 u'_2 a_2\cdots u'_m a_m v
\;\sim_{n+m}\;  u'_1 a_1 u'_2 a_2\cdots u'_m a_m v'
\:,
\end{align*}
by repeated uses of \Cref{lem-rich-and-sim}.
\end{proof}

\subsection{The small-subword theorem}
Our next result is used to prove lower bounds on the PT height of long
words. It will be used  repeatedly in the course
of this article.

For $k=1,2,\ldots$ define $f_k:\Nat\to\Nat$ by induction on $k$ with
\begin{xalignat}{2}
\label{eq-def-f1}
f_1(n) &= n \:, \\
\label{eq-def-fk}
f_{k+1}(n) &= \max_{0\leq m\leq n}\bigl(m f_k(n+1-m)+m+f_k(n-m)\bigr)\:.
\end{xalignat}

In the rest of the article, we shall simplify statements involving the
$f_k(n)$ bound by relying on the following bound (proved in
the Appendix):
\begin{equation}
\label{eq-bound-fkn}
f_k(n) \leq \Bigl(\frac{n+2k-1}{k}\Bigr)^k -1 < \Bigl(\frac{n}{k}+2\Bigr)^k\:.
\end{equation}

\begin{theorem}[Small-subword Theorem]
\label{SSTHM}
Let $k=|A|$. For all $u\in A^*$ and $n\in\Nat$ there exists some
$v\in A^*$ with $v\simonRstar u$ and such that $|v|\leq f_k(n)$.
\end{theorem}
\begin{proof}
By induction on $k$, the size of the alphabet.

With the base case, $k=1$, we consider a unary alphabet $A=\{a\}$ and $u$ is
$a^{|u|}$. Now $a^\ell\sim_n u$ iff $\ell=|u|<n$ or $n\leq
\min (\ell, |u|)$. So
taking $v=a^\ell$ for $\ell=\min(n,|u|)$ proves the claim.

When $k>1$ we consider the rich factorization
$u=u_1 a_1 u_2 a_2\cdots u_m a_m u'$ of $u$.
Let $n'=\max(n+1-m,1)$. Every $u_i$ is a word on the subalphabet
$A\setminus\{a_i\}$. Hence by induction hypothesis there exists
$v_i\subword u_i$ with $|v_i|\leq f_{k-1}(n')$ and $v_i\sim_{n'}u_i$,
entailing $u_i a_i\sim_{n'}v_i a_i$. Similarly, the induction hypothesis
entails the existence of some $v'\subword u'$ with $v'\sim_{n'-1}u'$
and $|v'|\leq f_{k-1}(n'-1)$. Note that in these inductive steps we use a length
bound obtained with $f_{k-1}$ by using the fact that
$u_1,\ldots,u_m$ and $u'$, being poor, use at most $k-1$ letters from
$A$.

We now consider two cases. If $m \leq n-1$, we let $v=v_1
a_1\cdots v_m a_m v'$, so that $v\subword u$ and $|v|\leq m
f_{k-1}(n')+m+f_{k-1}(n'-1)$.  We deduce $|v|\leq f_k(n)$ using
Eq.~\eqref{eq-def-fk} and since $n'=n+1-m$.  That $v\sim_n u$, hence
$v\simonRstar u$, is an application of Lemma~\ref{lem-rich-middle}:
$v_1a_1\cdots v_ma_mv'$ is indeed the rich decomposition of $v$ since
$n' \geq 2$, $v'\sim_{n'-1}u'$, and $v_i\sim_{n'}u_i$ for $i=1,\ldots,m$.

If $m\geq n$, then $u$ is $n$-rich and its subwords include all words
of length at most $n$. It is easy to extract some $n$-rich subword $v$ of
$u$ that only uses $kn$ letters. Now $v\sim_n u$ since both $u$ and $v$  are $n$-rich, entailing
$v\simonRstar u$. One also checks that $|v|=kn\leq f_k(n)$.
\end{proof}
Note that the bound $f_k(n)$ in
Theorem~\ref{SSTHM} does not depend on $u$.  \\

We can already apply the small-subword theorem to the case of finite
languages.
\begin{proposition}[Finite languages]
\label{prop-ptl-finite-L}
Suppose $L\subseteq A^*$ is finite and nonempty with $|A|=k$. Let
$\ell$ be the length of the longest word in $L$. Then
$k(\ell+1)^{1/k}-2k+1\leq h(L)\leq \ell+1$.
\end{proposition}
\begin{proof}
Thanks to Eq.~\eqref{eq-h-finite}, it is enough to consider the case
where $L=\{u\}$ is a singleton. So assume $h(L)=h(u)=n$ and
$|u|=\ell$. The small-subword theorem says that $u\sim_n v$ for some
short $v$ but necessarily $v=u$ since $[u]_n$ is a singleton, hence
$\ell\leq f_k(n)$. Using Eq.~\eqref{eq-bound-fkn} one gets $\ell\leq
f_k(n)\leq \bigl(\frac{n+2k-1}{k}\bigr)^k-1$. This gives
$n\geq k(\ell+1)^{1/k}-2k+1$ as announced. The upper bound $h(L)\leq \ell+1$ was
observed earlier.
\end{proof}
\begin{remark}[On Tightness]
We already noted that the $\ell+1$ upper bound is tight. The lower bound is
quite good: for $U_k$ seen above, $\ell=(\lambda+1)^k-1$, so that
$\ell\leq\bigl(\frac{n+2k-1}{k}\bigr)^k-1$ gives $n=h(U_k)\geq k\lambda-k+1$
while we know $h(U_k)=k\lambda+1$.
\qed
\end{remark}

Finding tight
bounds for the trade-off between word length and PT-height
 is an interesting open problem.  The
existing gap in \Cref{prop-ptl-finite-L} can be narrowed at one end by
improving the small-subword theorem and, at the other end, by
discovering words with small PT-height as a function of their
length. In this direction, we note that our $U_k$ words provably do
not hold the record: for example, for $k=3$ and
$w=a^3b^4a^2c^4a^4cb^3c^3b^2=aaabbbbaaccccaaaacbbbcccbb$, we have
$|w|=26$ and $h(w) = 6$, to be compared with $|U_3|=26$ and $h(U_3)=7$
when $\lambda=2$.




\section{Upward closures}
\label{sec-PT-up}

Recall that $\up L$ is PT for any $L\subseteq A^*$.  Related
languages are $\up_< L$ (used in \Cref{sec-FO-basic})
and $\min (L) \egdef \{u \in L ~|~ \forall v\in
L:v\not\ssubword u\}$. This section provides bounds on the PT height of
these languages as a function of $L$.

We first note that, in the special case where $L$ is a singleton, the
PT height of
$\up L$ and $I(L)$ always coincide with word length:\footnote{This
phenomenon does not extend to the other operations nor to finite
sets.}
\begin{proposition}
\label{prop-ptl-up-cone}
For any $u \in A^*$
\begin{xalignat}{2}
\label{eq-ptl-up-I-singleton}
h(\up u) &= |u|\:,
&
h(I(u)) &= h(\up u\cup \down u)=\begin{cases} |u| &\text{if $|A|\geq 2$},\\
                            0 & \text{otherwise.}
\end{cases}
\end{xalignat}
\end{proposition}
\begin{proof}
Let $\ell=|u|$. Obviously $h(\up u)\leq\ell$ and the point is to
prove $h(\up u)>\ell-1$. For this we assume $\ell>0$ and write
$u=a_1\cdots a_\ell$. With each letter $a\in A$ we associate a word
$\pi_a$ of length $|A|$ that lists all the letters of $A$ exactly once
\emph{and ends with $a$}. E.g.\ $\pi_b=acdb$ works when
$A=\{a,b,c,d\}$. Let now $v=\pi_{a_1}\pi_{a_2}\cdots\pi_{a_{\ell-1}}$
and $v'=v\cdot a_\ell$. Then $v\sim_{\ell-1}v'$ since $v$ has all
subwords of length $\ell-1$. However $u\not\subword v$ and $u\subword
v'$ hence $\up u$ is not closed under $\sim_{\ell-1}$.

Now for $I(u)$, we 
note that
$h(I(u))\leq \max(h(\up u),h(\down_< u))$ since
$I(u)=A^*\setminus ( \up u \cup \down_< u)$, and that $\max(h(\up u),h(\down_< u))=\ell$ since
$h(\up u)=\ell$ and since all the finitely many words in
$\down_< u$ have length at most $\ell-1$.
 To show $h(I(u))>\ell-1$
when $|A|\geq 2$, we assume $\ell>1$ and use $v$ and $v'$ again: $v'\not\in I(u)$ while
$v\in I(u)$ hence $I(u)$ is not closed under $\sim_{\ell-1}$. Finally,
when $|A|<2$ or $\ell=0$, $I(u)=\emptyset$, while  when $\ell=1$ and
$|A|\geq 2$, $I(u)$ is neither $\emptyset$ nor $A^*$ so $h(I(u))>0$.
\end{proof}

\begin{corollary}
\label{coro-pt-upL}
For any $L\subseteq A^*$ and $m\in\Nat$, if all words in $\min(L)$
have length bounded by $m$, then $h(\up L)\leq m$ and $h(\up_< L)\leq
m+1$.
\end{corollary}
\begin{proof}

Since $\up L=\up\min(L)$ and since $\min(L)$ is finite (by Higman's
Lemma), we have  $h(\up L)=h(\bigcup_{u\in\min(L)}\up u)
\leq\max_{u\in\min(L)} h(\up u)
=\max_{u\in\min(L)} |u|\leq m$.

Now since $\up_< L=(\up L)\setminus\min(L)$, we deduce $h(\up_< L)\leq
\max (h(\up L),h(\min(L)))$. But $h(\min(L))\leq m+1$ by
\Cref{prop-ptl-finite-L}.
\end{proof}
This can be immediately applied  to languages given by automata or
 grammars.
\begin{theorem}[Upward closures of regular and context-free languages]
\label{thm-up-reg-cfg}~

\begin{enumerate}
\item
If $L$ is accepted by a nondeterministic automaton (a NFA) having
depth $m$, then $h(\up L)\leq m$ while $h(\up_< L)\leq m+1$ and
 $h(\min(L))\leq m+1$.

\item
The same holds if $L$ is accepted by a context-free grammar (a
CFG) when we let $m=\ell^N$ where $N$ is the number of nonterminal
symbols and $\ell$ is the maximum length for the right-hand side of
production rules.
\end{enumerate}
\end{theorem}
\begin{proof}
(1) A word accepted by the NFA is minimal wrt $\subword$ only if it is
accepted along an acyclic path.
(2) A word generated by the CFG is minimal wrt $\subword$ only if any
nonterminal appears at most once along any branch of its smallest
derivation tree.
\end{proof}
The bounds in \Cref{thm-up-reg-cfg} can be reached, e.g., for $L$ a singleton of the form $\{a^m\}$.
\\

For our applications, we are interested in bounding $h(\up L)$ in
terms of $h(L)$, assuming that $L$ is PT.
\begin{theorem}[Upward closures of PT languages]
\label{thm-ptl-for-upL}
Suppose that $L\subseteq A^*$ is PT and let $k=|A|$ and $m=f_k(h(L))$. Then
\begin{xalignat*}{3}
h(\up L)&\leq m \:,
&
h(\up_<L) &\leq m + 1 \:,
&
h(\min(L)) &\leq m + 1 \:.
\end{xalignat*}
\end{theorem}
\begin{proof}
By the small-subword theorem, and since $L$ is closed under $\sim_{h(L)}$,
the minimal elements of $L$ have length bounded by
$m$. Then Corollary~\ref{coro-pt-upL} applies.
\end{proof}
\begin{remark}
The upper bound in \Cref{thm-ptl-for-upL} is quite good: for any
$k,\lambda\geq 1$, the language $L=\{U_k\}$ has $h(L) = h(U_k)=k\lambda+1$ so
that \Cref{thm-ptl-for-upL} with Eq.~\eqref{eq-bound-fkn}
give $h(\up U_k)\leq f_k(k\lambda +1)\leq (\lambda+2)^k - 1$. On the other hand
we know that $h(\up U_k) = (\lambda+1)^k - 1$ by \Cref{prop-ptl-up-cone}.
\end{remark}




\section{Downward closures}
\label{sec-PT-down}

We now move to downward closures.  It is known that, for any
$L\subseteq A^*$, $\down L$ and $\down_< L$ are PT since they are the
complement of upward-closed languages.  Our strategy for bounding
$h(\down L)$ is to approximate $L$ by finitely many D-products.

\begin{definition}
A \emph{D-product over $A$} is a regular expression $P$ of the form $E_1\cdot
E_2\cdots E_\ell$ where every $E_i$ is either of the form $B^*$ for a
subalphabet $B\subseteq A$ ($B^*$ is called a \emph{star factor} of
$P$), or a single letter $a\in A$ (a \emph{letter factor}). We say
that $\ell$ is the length of $P$.
\end{definition}
As is common, we abuse notation and let $P$ denote both a regular
expression and the associated language.

We note that our D-products are slightly more general than the
monomials of the form $B_0^*a_1 B_1^*a_2 \cdots a_n B_n^*$ considered
in~\cite{DGK-ijfcs08}, where a strict alternation is imposed between
star factors and letter factors. However, any D-product is easily
translated as a polynomial (a finite sum of monomials) by replacing
any two consecutive letter factors $a\cdot a'$ by the equivalent
$a\cdot \emptyset^*\cdot a'$, and any two consecutive star factors $B^*
\cdot B^{\prime *}$ by $B^*+\sum_{a\in B'}\bigl(B^*\cdot a\cdot B^{\prime *}\bigr)$ and then
distributing concatenations over unions. Thus the languages described
by finite unions of D-products are exactly the languages described by
polynomials (see~\cite{DGK-ijfcs08,pin97} for algebraic and logical
characterisations).  \\

D-products generalise words, and they share with words their nice
upper bound on the PT height of downward closures:
\begin{proposition}
\label{prop-ptl-of-sre}
Let $P$ be a D-product of length $\ell$. Then $h(\down P)\leq \ell+1$
and $h(\down_< P)\leq \ell+1$.
\end{proposition}
\begin{proof}
Let $P'$ be the regular expression obtained from $P$ by replacing any
letter factor $a$ by $(a+\epsilon)$ so that $P'=\down P$.  We claim
that any residual $w^{-1}P''$ of a suffix $P''$ of $P'$ is either the
empty language $\emptyset$, or is itself a suffix of $P'$. The claim
is proven by induction on the length of $P''$, then on the length of
$w$, recalling that residuals can be computed inductively via
$\epsilon^{-1} L=L$ and $(w b)^{-1}L=b^{-1}(w^{-1}L)$. When
considering suffixes of $P'$ (or $\emptyset$), the following
equalities can be used:
\begin{xalignat*}{2}
b^{-1}\epsilon &=\emptyset,
&
b^{-1}\emptyset &= \emptyset,
\\
b^{-1} \bigl [(a+\epsilon)P''\bigr] &=\left\{\begin{array}{ll}P'' &\text{if }b=a,\\
                                             b^{-1} P'' &\text{otherwise},
                            \end{array}\right.
&
b^{-1} \bigl[B^*P''\bigr]        &=\left\{\begin{array}{ll}B^*P'' &\text{ if }b\in B,\\
                                             b^{-1}P'' &\text{otherwise}.
                            \end{array}\right.
\end{xalignat*}
Note that the correctness of the third equality when $b=a$, and of the
fourth equality when $b \in B$, rely on $b^{-1}P'' \subseteq P''$:
this holds because $P'$, and then each suffix $P''$, is
downward-closed.

Finally, $P'$ has at most $\ell+1$ distinct non-empty residuals since
it has $\ell+1$ suffixes. Thus the minimal DFA for $P'$ has at most
$\ell+1$ productive states, hence has depth at most $\ell+1$. We now
apply Theorems 1 and 2 from \cite{klima2013} and conclude that
$h(\down P)\leq\ell+1$.

For bounding $h(\down_< P)$ very little is changed. If $P$ contains at
least one (nonempty) star factor then $\down_< P$ and $\down P$
coincide. If $P$ only contains letter factors (and empty star factors)
then $P$ denotes a singleton $\{u\}$ with $|u|\leq\ell$ and $\down_<
P$ is a finite set of words of length at most $\ell-1$, entailing
$h(\down_< P)\leq\ell$.
\end{proof}
The bounds in \Cref{prop-ptl-of-sre} can be reached, e.g., for
$P=a\cdots a$.
\begin{corollary}
\label{coro-ptl-of-sre}
If $L\subseteq \bigcup_i P_i\subseteq \down L$ for a family  $(P_i)_i$ of
D-products of length at most $\ell$, then
$h(\down L)\leq \ell+1$ and
$h(\down_< L)\leq \ell+1$.
\end{corollary}
\begin{proof}
Obviously $\down L=\bigcup_i \down P_i$ and $\down_< L=\bigcup_i
\down_< P_i$.  These unions are finite since there are only finitely
many D-products of bounded length, so that we can invoke
Eq.~\eqref{eq-basic-1}.
\end{proof}
This can be immediately applied to languages given by automata or
grammars.
\begin{theorem}[Downward-closures of regular and context-free languages]
~

\begin{enumerate}
\item
If $L$ is accepted by a nondeterministic automaton (a NFA) having
depth $m$, then $\down L$ and  $\down_< L$ are $\ell$-PT
for $\ell=2m+2$.

\item
The same holds if $L$ is accepted by a CFG in quadratic normal
form (a QNF, see~\cite{bachmeier2015}) with $N$ nonterminals
and $\ell=4\cdot 3^{N-1}+2$.
\end{enumerate}
\end{theorem}
\begin{proof}
(1) For a word $u\in L$ we consider the cycles in an accepting path on
$u$. This leads to a factoring $u=u_0a_1 u_1 a_2 \cdots a_p u_p$ of
$u$ such that the accepting path is some
$q_0\step{u_0}q_0\step{a_1}q_1\step{u_1}q_1\step{a_2}q_2 \cdots
q_{p-1}\step{a_p}q_p\step{u_p}q_p$ with $q_0,q_1,\ldots,q_p$ all
different.  Then $p\leq m$. Let now $B_i\subseteq A$ be the set of
letters occurring in $u_i$ and define $P_u\egdef
B_0^*a_1B_1^*a_2\ldots B_{p-1}^*a_pB_p^*$. Then $u\in P_u$ and
$P_u\subseteq\down L$.  Finally, $L\subseteq \bigcup_{u\in
L}P_u\subseteq\down L$ and each $P_u$ has length $\leq 2m+1$. One
now invokes \Cref{coro-ptl-of-sre}.

\noindent
(2) Bachmeier \etal showed that there is an NFA for $\down L$ having
 $2\cdot 3^{N-1}$ states~\cite{bachmeier2015}.
\end{proof}
For our applications, we are interested in bounding $h(\down L)$ in
terms of $h(L)$ when $L$ is PT. Our main result in this section is:
\begin{theorem}[Downward closures of PT languages]
\label{thm-ptl-for-downL}
Suppose that $L\subseteq A^*$ is PT and let $k=|A|$ and
 $m=f_k(h(L))$. Then
\begin{xalignat*}{2}
h(\down L)&\leq (k+1)(m+1) \:,
&
h(\down_< L)&\leq (k+1)(m+1) \:.
\end{xalignat*}
\end{theorem}

\begin{remark}
Before proving \Cref{thm-ptl-for-downL}, let us observe that the
upper bound it provides is quite good: for any $k,\lambda\geq 1$, the
language $L=\{U_k\}$ from \Cref{ssec-uk} has $h(U_k)=n=k\lambda+1$ so
that \Cref{thm-ptl-for-downL} gives $h(\down U_k)<
(k+1)(\lambda+2)^k$. On the other hand we know that $h(\down U_k)=
\lambda(\lambda+1)^{k-1}+1$ by Eq.~\eqref{eq-bounds-uk}.
\end{remark}

The rest of this section is devoted to the proof of
Theorem~\ref{thm-ptl-for-downL}.  Let $n=h(L)$ and $m=f_k(h(L))$. Our strategy is to cover $L$ by
D-products of bounded length, relying on the fact that $L$ is closed
under $\sim_n$.
\begin{lemma}
\label{lem-first-Pu}
For every $u\in A^*$ there is a D-product $P_u$ with at most $m$ letter
factors and such that $u\in P_u\subseteq [u]_n$.
\end{lemma}
\begin{proof}
By the small-subword theorem, $u$ has a  subword $v=a_1\cdots
a_\ell$  with $v\sim_n u$ and $|v|=\ell\leq m$. Thus $u$ has
the form
\begin{equation}
\label{eq-v-in-u}
u
\;=\;
u_0\,a_1 \,u_1\,a_2\,u_1 \cdots a_\ell\, u_\ell
\;=\;
b_{0,1} \cdots b_{0,p_0}\, a_1\, b_{1,1}\cdots
b_{1,p_1}\, a_2\cdots a_{\ell}\, b_{\ell,1}\cdots b_{\ell,p_{\ell}}
\:.
\end{equation}
Note that the above factorization is not necessarily unique, we just
fix one. Then the $b_{i,j}$'s are the letters making up the $u_i$ factors in
\eqref{eq-v-in-u}.
To shorten notation, we let $a_0$ stand for $\epsilon$ so
that we can write $u=\prod_{i=0}^\ell
a_iu_i=\prod_{i=0}^{\ell}(a_i\prod_{j=1}^{p_i}b_{i,j})$.

We claim that
$P_u\egdef\prod_{i=0}^{\ell}\bigl(a_i\prod_{j=1}^{p_i}\{b_{i,j}\}^*\bigr)$
proves the Lemma. That $u\in P_u$ and that $P_u$ has at most $m$
letter factors is clear. To show $P_u\subseteq[u]_n$, it is enough to
invoke a natural generalisation of Lemma~\ref{lem-useful}~(\ref{it-pumping}) that, for
the sake of completeness, we state and prove as
Lemma~\ref{lem-pumping-generalized}.
\end{proof}
\begin{lemma}
\label{lem-pumping-generalized}
Assume $u=u_0 u_1 u_2\cdots u_\ell\simonRstar u_0 a_1 u_1 a_2 u_2
\cdots a_\ell u_\ell=v$ for some words $u_0,\ldots,u_\ell$ and letters
$a_1,\ldots,a_\ell$.  Then $u_0 a_1^* u_1 a_2^* u_2 \cdots a^*_\ell
u_\ell\subseteq[u]_n=[v]_n$.
\end{lemma}
\begin{proof}
By induction on $\ell$.  The case $\ell =0$ is trivial so we assume
$\ell>0$.  From $u\simonRstar v$, and via Lemma~\ref{lem-useful}~(\ref{it-convex}), we
deduce $u'\simonRstar v$ for $u'= u_0 u_1 a_2 u_2 \cdots a_\ell
u_\ell$.  With $u'\simonRstar v$, Lemma~\ref{lem-useful}~(\ref{it-pumping}) gives
\begin{equation}
\label{eq-pumping-generalized}
u_0 a_1^* u_1 a_2
u_2 \cdots a_\ell u_\ell \subseteq [u]_n=[u']_n=[v]_n
\:.
\end{equation}
Pick $k\in\Nat$ and write $w_k$ for $u_0a_1^k u_1 a_2 u_2 \cdots
a_\ell u_\ell$. Eq.~\eqref{eq-pumping-generalized} entails $w_k \sim_n
u$. With $u\subword u_0a_1^ku_1u_2\cdots u_\ell\subword w_k$, and
since $u\sim_n v\sim_n w_k$, Lemma~\ref{lem-useful}~(\ref{it-convex}) now gives
$u_0a_1^ku_1u_2\cdots u_\ell\simonRstar w_k$. With the induction
hypothesis, we obtain
\begin{equation}
u_0 a_1^k u_1 a_2^* u_2 \cdots a_\ell^* u_\ell\subseteq
[w_k]_n=[u]_n=[u_n]
\:.
\end{equation}
Since this holds for any $k\in\Nat$, we have proved the Lemma.
\end{proof}
After bounding the letter factors in the $P_u$'s
(\Cref{lem-first-Pu}), we consider their star factors.  For this it is
convenient to write a D-product with $\ell$ letter factors under the
form $P= \prod_{i=0}^\ell \bigl( a_i\prod_{j=1}^{p_i} B_{i,j}^*
\bigr)$, i.e., regrouping the star factors in blocks separated by the
letter factors, and again with $a_0$ standing for $\epsilon$.  With
such a D-product, we associate $P'= \prod_{i=0}^\ell a_i \bigl(
\prod_{j=1}^{p_i} B_{i,j}^{\prime *}\bigr)$ with the $B'_{i,j}$'s
given by
\begin{equation}
\label{eq-def-B'ij}
B'_{i,j} \egdef \bigl(B_{i,1}\cup B_{i,2}\cup\cdots\cup B_{i,j}\bigr)\cap
\bigl(B_{i,j}\cup B_{i,j+1}\cup\cdots\cup B_{i,p_i}\bigr)
\:.
\end{equation}
That is, any star product $B_{i,j}$ is enlarged with letters that
occur both on its left (inside the block of star products) and on its
right.  For example, with
\begin{align}
\notag
P_0&=d^* \, a_1 \, b^* \, (c+c')^* \, d^* \, (b+e)^* \, c^* \, a_2\, e^*
\:,
\\
\shortintertext{we associate}
\tag{$\dagger$}
\label{eq-P'_0-1}
P'_0 &= d^* \, a_1 \, b^* \, (b+c+c')^* \, (b+c+d)^* \, (b+e+c)^* \, c^* \, a_2 \, e^*
\\
\tag{$\ddagger$}
\label{eq-P'_0-2}
& \equiv d^* \, a_1 \, (b+c+c')^* \, (b+c+d)^* \, (b+e+c)^* \, a_2 \, e^*
\:.
\end{align}
Since $B_{i,j}\subseteq B'_{i,j}$ for all $i,j$, one has $P\subseteq
P'$. However $P'$ only enlarges $P$ in the following safe way:
\begin{lemma}
\label{lem-P-P'}
For any $u\in A^*$ and $n\in\Nat$, $P\subseteq [u]_n$ implies
$P'\subseteq [u]_n$.
\end{lemma}
\begin{proof}
A word in $P'$ may use letters from some $B'_{i,j}$'s that are not in
the corresponding $B_{i,j}$ and, with any $w\in P'$, we associate
$\#_w$, the smallest number of letters that must be removed from $w$
before the resulting subword belongs to $P$.

We now assume $P\subseteq [u]_n$ and prove $w\in [u]_n$ for all $w\in
P'$ by induction on $\#_w$. The base case where $\#_w=0$ is trivial
because then $w\in P$. So assume that $\#_w>0$, i.e., $w\not\in
P$. Then there must exist in $w$ an occurrence of some letter $b$,
from some $B'_{i,j}$, that does not appear in the corresponding
$B_{i,j}$. (Note that, by Eq.~\eqref{eq-def-B'ij}, there must exist
some $r<j$ and some $s>j$ such that $b\in B_{i,r}\cap B_{i,s}$.)
Accordingly, we factor $w$ under the form
\[
w = w_1 \, \underline{a_i}  \, \beta_1 \,  \underline{b} \,
\beta_2 \, \underline{a_{i+1}} \, w_2
\:,
\]
highlighting the selected $b$, the occurrences of $a_i$ and $a_{i+1}$
that surround it, and where $\beta_1$ and $\beta_2$ belong to
$B_{i,1}^{\prime *}\cdots B_{i,j}^{\prime *}$ and $B_{i,j}^{\prime
*}\cdots B_{i,p_i}^{\prime *}$, respectively.  Let
$w'=w_1\,a_i\,\beta_1\,\beta_2\,a_{i+1}\,w_2$, i.e., $w'$ is $w$
without the $b\in B'_{i,j}$ that we singled out, so that
$\#_{w'}<\#_{w}$ and the induction hypothesis yields $w'\in [u]_n$.
We now claim that $w\sim_n w'$.  To prove this, and since $w'\subword
w$, it is enough to show that any subword $t\in A^{\leq n}$ of $w$ is
also a subword of $w'$. So consider one such $t$.  From $t\subword w$
we extract a factorization $t=t_1\,t_2\, t_3$ such that
\begin{xalignat}{3}
\label{eq-decomp-t123}
t_1 &\subword w_1 \, a_i  \, \beta_1
\:,
&
t_2&\subword b
\:,
&
t_3
&\subword \beta_2 \, a_{i+1} \, w_2
\:.
\end{xalignat}
Thus $t_1t_3\subword w'$ and, since $w'\in[u]_n$ and $|t_1t_3|\leq n$,
$t_1t_3\subword x$ for any $x\in [u]_n$.  In particular,
$t_1t_3\subword a_1 a_2 \cdots a_\ell\in P$, which requires
$t_1\subword a_1 a_2\cdots a_i$ or $t_3\subword a_{i+1} a_{i+2}\cdots
a_\ell$ (or both).  Let us assume $t_1\subword a_1 a_2 \cdots a_i$,
the other case being similar.  Combining with Eq.~\eqref{eq-decomp-t123},
we obtain
\[
t= t_1t_2t_3 \subword (w''\egdef)\: a_1 a_2 \cdots a_i \cdot b \cdot \beta_2 a_{i+1} w_2
\:.
\]
Note that $w''\in P'$ and that $\#_{w''}<\#_{w}$ since the $b$ that
follows $a_i$ in $w''$ can be accounted for by $B_{i,r}^*$.  Thus
$w''\in [u]_n$ by induction hypothesis, i.e.\ $w''\sim_n w'$, from
which we deduce $t\subword w'$.
\end{proof}

We can now bound the number of star factors in the D-product $P'$
associated with $P$.  For this, we first simplify
$P'=\prod_{i=0}^\ell\bigl(a_i\prod_{j=1}^{p_i}B_{i,j}^{\prime
*}\bigr)$ by removing any $B_{i,j}^{\prime *}$ star factor that is
subsumed by its immediate neighbour, i.e., such that
$B'_{i,j}\subseteq B'_{i,j-1}$ or $B'_{i,j}\subseteq B'_{i,j+1}$. This
is exactly how we moved from \eqref{eq-P'_0-1} to \eqref{eq-P'_0-2} in
our earlier example, and it shortens $P'$ without changing the denoted
language. Once no more simplifications are possible, we can bound the
length of the resulting D-product with the following combinatorial
observation:
\begin{lemma}
\label{lem-bound-size-P'}
Assume that $A_{1},\ldots,A_{p}\subseteq A$ are $p$ subalphabets such
that
\\
--- for all $1 \leq j < p$,  $A_{j}\not\subseteq A_{j+1}$  and $A_{j+1}\not\subseteq A_j$ ;
\\
--- for all $b \in A$ and $1 \leq j <k < j' \leq p$, if $b \in
  A_{j}\cap A_{j'}$, then $b \in A_{k}$.
\\
Then $p \leq |A|$.
\end{lemma}
\begin{proof}
Note that by the first condition, each $A_{j}$ is nonempty. Extend
the sequence by defining $A_{0} = A_{p+1} = \emptyset$. For $0
\leq j \leq p$, define $\Delta_j = A_{j} \bigtriangleup
A_{j+1}$, where $\bigtriangleup$ denotes symmetric difference. Now
$\Delta_0$ and $\Delta_{p}$ have size at least $1$, and by the first
condition, every other $\Delta_j$ has size at least $2$. Thus
$\sum_{j=0}^{p} |\Delta_j| \geq 2 p$. By the second condition, any
$b \in A$ occurs in at most two $\Delta_j$'s, thus $\sum_{j=0}^{p}
|\Delta_j| \leq 2|A|$. So we conclude $2 p \leq 2|A|$.
\end{proof}
\begin{corollary}
\label{coro-second-Pu}
For every $u\in A^*$ there is a D-product $P'_u$ of length at most $k
m + m + k$, and such that $u\in P'_u\subseteq [u]_n$.
\end{corollary}
\begin{proof}
$P'_u$ as constructed above (and after simplifications) has at most
$m$ letter factors, that separate at most $m+1$ blocks of star
factors. Each such block is based on subalphabets
$B'_{i,1},\ldots,B'_{i,p_i}$ that satisfy the assumptions of
\Cref{lem-bound-size-P'}: condition 1 holds since otherwise more
simplifications could be performed, while condition 2 is a consequence
of the definition of the $B'_{i,j}$'s via Eq.~\eqref{eq-def-B'ij}.
Consequently, any star factor block in $P'_u$ has length at most
$|A|=k$, leading to a $km+m+k$ bound for the total length of $P'_u$.
Finally, that $u\in P'_u\subseteq[u]_n$ is a consequence of
Lemma~\ref{lem-P-P'} since $u\in P_u\subseteq [u]_n$.
\end{proof}

We may now conclude the proof of \Cref{thm-ptl-for-downL}.  Indeed,
with \Cref{coro-second-Pu}, and since $L$ is closed under $\sim_n$, we
obtain $L= \bigcup_{u\in L}P'_u$, where each $P'_u$ has length bounded
by $k m+k+m$.  We then apply \Cref{coro-ptl-of-sre}.




\section{Piecewise-testability and PT height for $I(L)$}
\label{sec-PT-IL}

Recall that $I(L)$ is the set of words which are incomparable (via
$\subword$ or $\supword$) with \emph{some word} in $L$.  I.e., it is
the image of $L$ by the incomparability relation $\perp$, or
equivalently its pre-image since $\perp$ is symmetric.

In this section we prove the following result.
\begin{theorem}
\label{thm-I-of-PT-is-PT}
Suppose $L\subseteq A^*$ is PT and let $k=|A|$ and
$m=f_k(h(L))$.  Then $I(L)$ is PT and
\begin{equation*}
h(I(L)) \leq m+1
\:.
\end{equation*}
\end{theorem}
We saw  that $\perp_{A^*}$ is a rational
relation, so that  $I(L)$ is regular
when $L$ is regular, see  Proposition~\ref{perp-rational}.
Showing that $I$ also preserves piecewise-testability requires more
work. For such questions, $I$ does not behave as simply as the
pre-images we considered in earlier sections.  In particular it does
not necessarily yield languages that are PT, unlike $\up L$ or $\down
L$.

At this point it is useful to examine some examples and make some
general observations. Let $A=\{a,b,c\}$ and define the language $L_1$
of all finite prefixes of $(abc)^{\omega}$ via
\[
L_1 =
(abc)^*(\epsilon+a+ab) = \{\epsilon,a,ab,abc,abca,abcab,\ldots\}
\:.
\]
Note that $L_1$ is totally ordered by $\subword$ hence no word of
$L_1$ is in $I(L_1)$, i.e., $I(L_1)\subseteq A^*\setminus L_1$.

To prove the reverse inclusion, we rely on the fact that a word is
incomparable with any other word having same length.  I.e.,
$I(u)\supseteq A^{=|u|}\setminus \{u\}$ for any $u$, and thus, for any
language $L$,
\begin{equation}
\label{eq-prop-I}
I(L) = \bigcup_{u\in L}I(u)
\supseteq
\bigcup_{u\in L} \bigl(A^{=|u|}\setminus \{u\}\bigr)
\:.
\end{equation}
Since $L_1$ above contains at least one word of any given length,
Eq.~\eqref{eq-prop-I} entails $I(L_1)\supseteq A^*\setminus
L_1$. Finally we have proved that $I(L_1)=A^*\setminus L_1$.  Thus
$I(L_1)$ is not PT since $L_1$ is not.  \\

A similar example shows that $I(L)$ is not necessarily regular when $L$ is not.
For example, take $A = \{a,b\}$ and let
\[
L_2=\{a^\ell b^\ell(\epsilon+b)~|~\ell\in\Nat\}=\{\epsilon,b,ab,abb,aabb,a^2b^3,
a^3b^3,\ldots\}\:.
\]
Here too $L_2$ is totally ordered by $\subword$ and contains one word
of each length. Hence $I(L_2)=A^*\setminus L_2$, which is not regular.
\\

Let us now consider some PT languages. In the case of a singleton
language $L = \{w\}$, we know from Eq.~\eqref{eq-ptl-up-I-singleton}
that $h(I(w))=|w|$ when $|A|\geq 2$, and $h(I(w))=0$ when
$|A|<2$. This can be used to bound $h(I(F))$ for a finite language
$F$, using $I(F)=\bigcup_{w\in F} I(w)$.
\\

Consider now $L_3=[aab]_2$, i.e., $L_3=aaa^*b$. This language is
infinite but it is totally ordered by $\subword$ and has one word of
each length $\ell\geq 3$. Hence $I(L_3)=A^*\setminus L_3
\setminus\down aab$ and we can easily bound $h(I(L_3))$ using results
from the previous sections.
\\

Another infinite PT language is $L_4=[aabb]_2$, i.e.,
$L_4=aaa^*bbb^*$. A different strategy applies here: $L_4$ contains
no words of length $\ell<4$ and exactly $\ell-3$
words of each length $\ell\geq 4$.  We may invoke a consequence of
Eq.~\eqref{eq-prop-I}: if a language $L$ contains at least two words
having same length $\ell$ then $I(L)$ contains all words of length
$\ell$.  Applied to $L_4$, this entails $I(L_4)\supseteq A^{\geq 5}$,
which is enough to conclude that $I(L_4)$ is $A^{\geq 5}\cup F$ for
some finite $F\subseteq A^{\leq 4}$, entailing $h(I(L_4))\leq 5$.
\\

As the above examples suggest, it is useful to think of the ``layers'' $L\cap
A^{=\ell}=\{w\in L : |w|=\ell\}$ of $L$, and classify them into
\emph{empty}, \emph{singular}, or \emph{populous} layers, depending on
whether they contain 0, 1, or more words. Observe that if $L\cap
A^{=\ell}$ is populous then $I(L)\cap A^{=\ell}$ equals $A^{=\ell}$.

It is also useful to decompose PT-languages into the equivalence
classes that make them up.  Therefore, in the rest of this section, we
focus on some equivalence class $[w]_n\subseteq A^*$ where $n=h(L)$.
\\

A first observation is that the populous layers of $[w]_n$ propagate
upwards:
\begin{lemma}
\label{lem-two-upclosed}
Let $p\in\Nat$. If $[w]_n\cap A^{=p}$ is populous, then $[w]_n\cap A^{=p+1}$ is populous too.
\end{lemma}
\begin{proof}
Assume that $[w]_n$ contains two different words $u_1,u_2$ of length
$p$. Then $p>0$ and these words can be written under the form
$u_1=u_0av_1$ and $u_2=u_0bv_2$ where $u_0$ is their longest common
prefix and $a,b$ are two distinct letters occurring at the first
position where $u_1$ and $u_2$ differ.  Applying
Lemma~\ref{lem-useful}~(\ref{it-diff-let}) we deduce that $[w]_n$
contains either $u_0abv_2$ or $u_0bav_1$.  Let us assume, w.l.o.g.,
that $u_0bav_1\sim_n w$ since the other case is similar.  We now claim
that $u_0bbv_2\sim_n w$. Since $w\sim_n u_2\subword u_0bbv_2$, it is
enough to show that every subword $s$ of $u_0bbv_2$ of length at most
$n$ is also a subword of $u_2$. So let us pick any such $s$ and factor
it as $s=s_0 s_b s_2$ with $s_0\subword u_0$, $s_b\subword bb$, and
$s_2\subword v_2$, and with furthermore $s_0$ chosen longest possible.
If $s_b\leq b$ then $s\subword u_0b v_2=u_2$ and we are done, so we
assume $s_b=bb$. Let $s'=s_0bs_2$ and note that $s' \subword u_2$,
hence $s'\subword u_1$ since $u_1\sim_n u_2$ and $|s'|<|s|\leq n$.
Now $s'=s_0bs_2\subword u_0av_1=u_1$ requires $bs_2\subword a v_1$
since our choice of $s_0$ longest entails $s_0b\not\subword u_0$. This
gives $s=s_0bbs_2\subword u_0bav_1$, hence $s\subword u_2$ since
$u_0bav_1\sim_n u_2$.

We  have shown that $[w]_n$ contains $u_0bav_1$ and $u_0bbv_2$, both
words having length $p+1$.
\end{proof}

Populous layers also propagate downwards in the following sense:
\begin{lemma}
\label{lem-two-downclosed}
Let $p\geq 2$. If $[w]_n\cap A^{=p}$ is populous then $[w]_n\cap
A^{=p-1}$ is populous or $[w]_n\cap A^{=p-2}$ is empty (or both).
\end{lemma}
\begin{proof}
Assume that layer $p$ is populous and that layer $p-2$ is non
empty. If layer $p-2$ is populous, then layer $p-1$ is populous by
\Cref{lem-two-upclosed} and we are done. So assume that $[w]_n\cap
A^{=p-2}=\{x\}$ is singular.  Then
Lemma~\ref{lem-useful}~(\ref{it-shorter}) entails that $x\subword w'$
for all $w'\in [w]_n\cap A^{\geq p-1}$, and since layer $p$ is not
empty, Lemma~\ref{lem-useful}~(\ref{it-convex}) entails that layer $p-1$
too is not empty.  Pick $y\in [w]_n\cap A^{=p-1}$ and factor it as
$y=uau'$ such that $x=uu'$.  This entails $z=uaau'\in [w]_n\cap
A^{=p}$ by Lemma~\ref{lem-useful}~(\ref{it-pumping}). By assumption,
$[w]_n\cap A^{=p}$ is populous, hence contains another word $z'\not=
z$ and there is a word $y'\in [w]_n\cap A^{=p-1}$ with $x\subword
y'\subword z'$. If $y'\neq y$ we have proved that layer $p-1$ is
populous and we are done. Otherwise $y\subword z'$ and we may consider
the different possibilities for $y=uau'\subword z'\neq z$, knowing
that $|z'|=|y|+1$.  If $z'$ is either some $uabu'$ or $ubau'$ with
$b\neq a$ then we let $y''=ubu'$. From $x=uu'\subword y''\subword z'$
we deduce $y''\in [w]_n$ by Lemma~\ref{lem-useful}~(\ref{it-convex}).
If $z'$ is $vau'$ with $u\strictsubword v$ we let $y''=vu'$ and again
deduce $y''\in [w]_n$ from $x\subword y''\subword z'$. (Note that
$y''\neq y$ since $vau'=z'\neq z=uaau'$.)  If $z'$ is $uav'$ with
$u'\strictsubword v'$, the same reasoning applies to $y''=uv'$. In all
three cases $y''\neq y$ and $|y''|=|y|=p-1$, showing that $[w]_n\cap
A^{=p-1}$ is populous.
\end{proof}

With Lemmas~\ref{lem-two-upclosed} and~\ref{lem-two-downclosed}, we
see that almost all layers of $[w]_n$ are populous as soon as one
is. More precisely, when one layer is populous, either all (nonempty)
layers are populous, or all are except for the lowest (nonempty) one.
We already saw an example of the second situation with $L_4 = [aabb]_2
= aaa^*bbb^*$, and an example of the first situation is
$L_5=[abcab]_2$, where the lowest nonempty layer is $L_5\cap
A^{=5}=\{abcab,abcba,bacab,bacba\}$.
\\

We now consider the general case:
\begin{lemma}
\label{lem-I-of-class-is-PT}
$h(I([w]_n))\leq m+1$.
\end{lemma}
\begin{proof}
Recall that $[w]_n$ is a singleton or is infinite
(Lemma~\ref{lem-useful}~(\ref{it-sing-inf})).  We consider two cases.
\begin{enumerate}
\item Assume that $[w]_n=\{w\}$ is a singleton. Then $h(I(w))\leq |w|$
  by Eq.~\eqref{eq-ptl-up-I-singleton} and $|w| \leq m$ by the
  small-subword theorem. Hence $h(I([w]_n))\leq m$.

\item Assume that $[w]_n$ is infinite.  Let $u$ be a shortest word in $[w]_n$ and write $p$ for
  its length $|u|$. By the small-subword theorem, $p \leq m$. Since
  $[w]_n$ is infinite, and by Lemma~\ref{lem-useful}~(\ref{it-convex}),
  it contains at least one word of each length $\geq p$, hence
  $I([w]_n)\supseteq A^*\setminus[w]_n\setminus \down_< u$ by Eq.~\eqref{eq-prop-I}. There are
  two subcases.
\begin{enumerate}
\item
If $[w]_n$ is a total order under $\subword$, no layer is populous,
hence $I([w]_n)= A^*\setminus[w]_n\setminus \down_< u$.  Since
$h(\down_< u)=|u|\leq m$ and $h([w]_n)=n\geq m$, we obtain
$h(I([w_n]))\leq m$.

\item If $[w]_n$ is not a total order under $\subword$, all layers
  above $p$ are populous by Lemmas~\ref{lem-two-upclosed} and
  \ref{lem-two-downclosed}. If $u$ is the unique shortest word in
  $[w]_n$, we have $I([w]_n)= A^*\setminus\down u$, entailing
  $h(I([w]_n))\leq |u|+1\leq m+1$.  Otherwise layer $p$ is populous
  too, entailing $A^{\geq p}\subseteq I([w]_n)$, i.e.,
  $I([w]_n)=A^{\geq p}\cup F$ for some finite $F\subseteq A^{<p}$.  We
  deduce $h(I([w]_n))\leq p\leq m$ by Eq.~\eqref{eq-h-finite}.
\qedhere
\end{enumerate}
\end{enumerate}
\end{proof}

We may now conclude:
\begin{proof}[Proof of \Cref{thm-I-of-PT-is-PT}]
Being $n$-PT, $L$ is a finite union $[w_1]_n \cup \cdots \cup
[w_\ell]_n$ of equivalence classes of $\sim_n$, so that
$I(L)=I([w_1]_n)\cup \cdots\cup I([w_\ell]_n)$.  Now each $I([w_i]_n)$
is $(m+1)$-PT by Lemma~\ref{lem-I-of-class-is-PT} so that $I(L)$ is
too.
\end{proof}

\begin{remark}
The upper bound in \Cref{thm-I-of-PT-is-PT} is quite good: for any
$k,\lambda\geq 1$, the language $L=\{U_k\}$ from \Cref{ssec-uk}
has $h(U_k)=n=k\lambda+1$ so that \Cref{thm-I-of-PT-is-PT} gives
$h(I(U_k))\leq (\lambda+2)^k$. On the other hand we know
by Eq.~\eqref{eq-ptl-up-I-singleton} that
$h(I(U_k))=|U_k|= (\lambda+1)^k-1$ when $k>1$.
\end{remark}




\section{Deciding the two-variable logic of subwords}
\label{sec-FO-basic}

In this section we use our results on PT heights to establish
complexity bounds on a decidable fragment of $\FO(A^*,\subword)$, the
first-order logic of subwords.

We assume familiarity with basic notions of first-order logic as
exposed in, e.g.,~\cite{harrison2009}: bound and free occurrences of
variables, quantifier depth of formulae, and fragments $\FO^n$ where
at most $n$ different variables (free or bound) are used.  In
particular, if $\phi(x_1,\ldots,x_n)$ has $n$ free variables, we write
$R_\phi$ for the $n$-ary relation defined by $\phi$ on the underlying
structure.

The signature of the $\FO(A^*,\subword)$ logic only contains one
predicate symbol, ``$\subword$'', denoting the subword relation.
Terms are variables taken from a countable set $X=\{x,y,z,\ldots\}$
and all words $w_1,w_2,\ldots\in A^*$ as constant symbols (denoting themselves).
For example, with $A=\{a,b,c,\ldots\}$, $\exists x \bigl(ab\subword
x\land bc\subword x\land \neg(abc\subword x)\bigr)$ is a true sentence
as witnessed by $x\mapsto bcab$.
\\

The logic of the subword relation is a logic of substructure ordering
like those considered by Je{\v{z}}ek and McKenzie
(see~\cite{jezek2009i} and subsequent papers). It is one of the
simplest and most natural substructure ordering occurring in computer
science~\cite{kuske2006}. In its full generality, this logic is
computably isomorphic with $\FO(\Nat,+,\times)$, hence
undecidable~\cite{kudinov2010}.  We showed that already the $\Sigma_2$
fragment is undecidable~\cite{KS-fosubw} and recently Halfon \etal
showed that even the $\Sigma_1$ fragment is
undecidable~\cite{HSZ-lics2017}. This was very surprising: by comparison,
``words equations'', i.e., the $\Sigma_1$ fragment of
$\FO(A^*,\cdot,=)$ in which the prefix relation can be defined, are
decidable in $\PSPACE$~\cite{diekert2002,plandowski2004,jez2016}.  \\

We have previously shown that  $\FO^2(A^*,\subword)$, the 2-variable fragment, is decidable by a
quantifier elimination technique~\cite{KS-fosubw}. In this article
we extend our earlier analysis of the expressive power and complexity
of the $\FO^2$.

When performing quantifier elimination, it is
convenient to enrich the basic logic by allowing all regular languages
$L_1,L_2,\ldots\in\Reg(A^*)$ as monadic predicates with the expected
semantics, and we shall temporarily adopt this extension.  We write
$x\in L$ rather than $L(x)$ and assume that $L$ is given via a regular
expression or a finite automaton --- For example, we can state that
$(a+b)^*$ is the downward closure of $(ab)^*$ with $\forall
x\bigl[x\in (a+b)^*\iff \exists y(y\in (ab)^*\land x\subword y)
\bigr]$.

\subsection{Subword-recognizable relations}
In order to characterise the $\FO^2$-definable relations, we need some
definitions.  A relation $R\subseteq A^*\times A^*$ is
\emph{subword-recognizable}, if it belongs to the boolean closure of
$\Rec(A^*\times A^*)\cup\{\subA,\supA\}$. It is
furthermore \emph{subword-piecewise-testable}, if it belongs to the
boolean closure of $\PT(A^*\times
A^*)\cup\{\subA,\supA\}$. We write
$\Rec_{\subword}(A^*\times A^*)$ and $\PT_{\subword}(A^*\times A^*)$
for the corresponding classes.

\begin{proposition}[Normal form for $\Rec_{\subword}(A^*\times A^*)$ and $\PT_{\subword}(A^*\times A^*)$]
A relation $R\subseteq A^*\times A^*$ is subword-recognizable if, and
only if, it can be written under the form
\begin{equation}
\label{eq-NF}
\tag{NF}
R \;=\; (\ssubA\cap R_1)
\cup (\idA\cap R_2)
\cup (\ssupA\cap R_3)
\cup (\perpA\cap R_4)
\end{equation}
for some recognizable relations $R_1,R_2,R_3,R_4$.

Furthermore, $R$ is subword-piecewise-testable if, and only if,
the relations $R_1,R_2,R_3,R_4$ can be chosen among the
piecewise-testable relations.
\end{proposition}
\begin{proof}
That the normal forms are subword-recognizable is clear since
$\idA$ and $\ssubA$ belong to the boolean closure of
$\{\subA,\supA\}$: they are
$\subA\cap\supA$ and
$\subA\setminus\idA$.

Showing the other direction, i.e., that a subword-recognizable or
subword-piecewise-testable relation $R$ can be put in normal forms,
is done by induction on the boolean combination realising $R$ from the
generators of the boolean closure. Here it suffices to show that normal forms are closed under
intersections and complementations.  Let us write
$\{\xi_1,\xi_2,\xi_3,\xi_4\}$ for $\{\ssubA, \idA, \ssupA, \perpA\}$
and $\bigcup_{i=1}^4\xi_i\cap R_i$ for normal forms.  Since the
$\xi_i$'s are pairwise disjoint, we have
\begin{gather}
\label{eq-union-subnpt}
\Bigl(\bigcup_{i=1}^4\xi_i\cap R_i\Bigr)
\cap
\Bigl(\bigcup_{i=1}^4\xi_i\cap R'_i\Bigr)
\;=\;
\bigcup_{i=1}^4\xi_i\cap (R_i\cap R'_i)
\:.
\end{gather}
Since the $\xi_i$'s  form a partition of $A^*\times A^*$, we
further have
\begin{gather}
\label{eq-complement-subnpt}
(A^*\times A^*)
\setminus (\xi_i\cap R)
\;=\;
\Bigl(\xi_i\cap\bigl[(A^*\times A^*)\setminus R\bigr]\Bigr)
\cup
\bigcup_{j\neq i}\xi_j\cap(A^*\times A^*)
\:.
\end{gather}
Thus we see that normal forms are closed under
boolean operations since $\Rec(A^*\times A^*)$ is.

Finally, the same proof applies to $\PT_{\subword}(A^*\times A^*)$.
\end{proof}

\subsection{Quantifier elimination for $\FO^2(A^*,\subword,L_1,L_2,\ldots)$}

We may now characterise the $\FO^2$-definable relations.
\begin{theorem}
\label{thm-fo2-rat}
(i) A relation $R\subseteq A^*\times A^*$ is definable in the extended
logic $\FO^2(A^*,{\subword,}$\linebreak $L_1,L_2,\ldots)$ iff it is subword-recognizable.

(ii) It is definable in the basic logic $\FO^2(A^*,\subword,w_1,w_2,\ldots)$
iff it is subword-piecewise-testable.

(iii) Furthermore, a normal form for $R_\phi$ can be computed from the
$\FO^2$ formula $\phi(x,y)$,
\end{theorem}
\begin{proof}
The $(\Leftarrow)$ direction of \textit{(i)} is obvious: for example
$R=\ssubA\cap L\times L'$ is definable via $x\subword y\land x\not=y
\land x\in L\land y\in L'$, a $\FO^2$ formula. When proving the
same $(\Leftarrow)$ direction for \textit{(ii)}, we cannot use
regular predicates to express $x\in L$ or $y\in L$. But since we
assume that $L$ is PT, it is a boolean combination of filters $\up u$,
$\up u'$, etc., so $x\in L$ can be expressed as a boolean
combination of atomic formulas $u\subword x$, $u'\subword x$, etc.

We prove the $(\Rightarrow)$ direction for \textit{(i-ii)} by structural induction on the
$\FO^2$-formula $\phi(x,y)$ that defines $R_\phi$.  We consider all
cases:

If $\phi$ is an atomic formula of the form $x\in L$ or $y\in L'$, then
$R_\phi$ is $L\times A^*$ or $A^*\times L'$. If $\phi$ is some $u\in
L$, then $R_\phi$ is one of the trivial $\emptyset$ or $A^*\times
A^*$.

If $\phi$ is an atomic formula of the form $x\subword u$ or $v\subword
y$ for some constant word $u$ or $v$ then $R_\phi$ is $(\down u)\times
A^*$ or $A^*\times (\up v)$ respectively (a piecewise-testable
relation in each case). If $\phi$ is $x\subword y$ then $R_\phi$ is
$\subA$, and $R_\phi$ is trivial if $\phi$ is some $u\subword v$.

If $\phi$ is a conjunction $\phi_1\land\phi_2$ or a negation
$\neg\phi_1$, we rely on the induction hypothesis and the closure
properties of $\Rec_{\subword}(A^*\times A^*)$
and  $\PT_{\subword}(A^*\times A^*)$.

The remaining case is when $\phi(x,y)$ is some $\exists x.\psi(x,y)$,
the case $\exists y.\psi(x,y)$ being identical.  Here $R_\phi =
A^*\times \pi_2(R_\psi)$ where $\pi_2:A^*\times A^*\to A^*$ is the
projection $(u,v)\mapsto v$ lifted from pairs to sets of pairs, i.e.,
relations.  To compute $\pi_2(R_\psi)$ we write $R_\psi$ in normal
form ---thanks to the induction hypothesis--- and use $\pi_2
\bigl(\bigcup_{i=1}^4\xi_i\cap R_i\bigr) = \bigcup_{i=1}^4
\pi_2\bigl(\xi_i\cap R_i\bigr)$ and the following equalities (proofs
omitted):
\begin{xalignat}{2}
\label{eq-elim-=-<}
\pi_2(=_{A*}\cap L\times L') &= L\cap L'
\:,
&
\pi_2(\ssubword_{A*}\cap L\times L') &= (\up_< L) \cap L'
\:,
\\
\label{eq-elim-perp->}
\pi_2(\perp_{A*}\cap L\times L') &= I(L) \cap L'
\:,
&
\pi_2(\ssupword_{A*}\cap L\times L') &= (\down_< L) \cap L'
\:.
\end{xalignat}
Finally, in the case where $\phi$ does not use regular predicates, the
above inductive construction only produces subword-piecewise-testable
relations.

We now see why \textit{(iii)} holds: the operations used above can all
computed effectively on relations in normal form, e.g., using a
quadruplet of automata for the $R_1,R_2,R_3,R_4$ of \eqref{eq-NF}.
The required operations on automata are classic constructions: boolean
combinations and images of regular languages by $\ssupword_A$,
$\ssubword_A$ and $\perp_A$ (all rational relations).
\end{proof}
We note that if $\phi(x)$ is a $\FO^2$ formula with a single free
variable, $R_\phi$ can be put under the form $L\times A^*$, i.e.,
$\phi(x)$ defines a regular property of words, and a
piecewise-testable one if $\phi(x)$ is in the basic logic.
\\

\Cref{thm-fo2-rat} has several corollaries.
Firstly, and since the
normal forms can be effectively computed, we have
\begin{corollary}[Decidability~\cite{KS-fosubw}]
Validity and satisfiability are decidable for
$\FO^2(A^*,\subword,\linebreak L_1,L_2,\ldots)$.
\end{corollary}
By contrast, note that the $\FO^3\cap\Sigma_2$ and the $\Sigma_1$
fragments of the basic logic are
undecidable~\cite{KS-fosubw,HSZ-lics2017}.

Secondly, the computations inside the proof of \Cref{thm-fo2-rat} can
be seen as a quantifier-elimination procedure.
\begin{corollary}
Any $\FO^2(A^*,\subword,L_1,L_2,\ldots)$ formula is effectively equivalent to a
quantifier-free $\FO^2$ formula. The same holds for the basic logic
$\FO^2(A^*,\subword,w_1,w_2,\ldots)$.
\end{corollary}

\subsection{Complexity for $\FO^2(A^*,\subword)$}
The algorithm underlying the proof of \Cref{thm-fo2-rat}~\textit{(iii)} can be
implemented using finite-state automata to handle and compute
subword-recognizable relations via their normal forms.
The steps described in
Eqs.~(\ref{eq-elim-=-<}--\ref{eq-elim-perp->}) involve computing
images of regular languages by (fixed) rational relations and may
induce an exponential complexity blow-up. In particular the
pre-images $\up_< L$ and $\down_< L$ can have exponential size if one
uses deterministic  or alternating automata, while if one uses
nondeterministic
(or unambiguous) automata,
the dual pre-images $\neg(\up_<(\neg L))$ and
$\neg(\down_<(\neg L))$ ---used in eliminating a universal quantifier--- can have doubly exponential size~\cite{KNS-tcs2016}.  Therefore
the best known upper bound for the decidability of
$\FO^2(A^*,\subword,L_1,L_2,\ldots)$ is a tower of exponentials with
height bounded by the nesting depth of the formula at hand, hence
a nonelementary complexity.  Regarding lower bounds, only
$\PSPACE$-hardness has been established~\cite{KS-fosubw}.  \\

We now turn to the basic logic, $\FO^2(A^*,\subword,w_1,w_2,\ldots)$ where
regular predicates are not allowed.  As stated in \Cref{thm-fo2-rat},
the quantifier-elimination procedure will only produce
subword-piecewise-testable relations and languages.  Furthermore, it
is possible to bound the PT height of the defined languages and deduce
an elementary complexity upper bound.

\begin{theorem}[$\FO^2(A^*,\subword,w_1,w_2,\ldots)$ has elementary complexity]
\label{maintheorem}
Assume that $\phi$ is a $\FO^2(A^*,\subword,w_1,w_2,\ldots)$ formula.
Then $h(R_\phi)$ is in $2^{2^{O(|\phi|)}}$.  \\
Furthermore, computing  DFAs for the normal form of $R_\phi$ (hence deciding the
satisfiability or the validity of $\phi$) can be
done in $\EXPTIMEiii$.
\end{theorem}
\begin{proof}
The quantifier-elimination procedure that proves
\Cref{thm-fo2-rat}~\textit{(ii)} builds, for any subformula $\psi$ of $\phi$, a
relation of the form $\bigcup_{i=1}^4 \xi_i\cap R_i$
represented by a quadruplet $R_1,R_2,R_3,R_4$ of
PT relations. The PT height of these relations can be bounded. For
example, the PT height is given by $|u|$ for $\psi$ an atomic formula
of the form $u\subword x$.  Boolean combinations do not increase
PT height even in the case of subword-piecewise-testable relations,
see Eqs.\ \eqref{eq-union-subnpt} and \eqref{eq-complement-subnpt}.
Quantifier-elimination can increase PT height when we
compute $\up_< L$, $\down_< L$ and $I(L)$ as prescribed by
Eqs.\ \eqref{eq-elim-=-<} and \eqref{eq-elim-perp->}. But
Theorems~\ref{thm-ptl-for-upL}, \ref{thm-ptl-for-downL} and
\ref{thm-I-of-PT-is-PT} apply and show that the increase is
polynomially bounded. Such increases combined at most $|\phi|$ times
give a PT height bounded in  $2^{2^{O(|\phi|)}}$.

Finally, when the PT height of $R_\phi$ (and of all intermediary
$R_\psi$) have been bounded in $2^{2^{O(|\phi|)}}$, we obtain a bound
on the size of the minimal DFAs and the time and space needed to
compute them using \Cref{thm-size-DFA-for-nPT}.
\end{proof}


\section{Concluding remarks}
\label{sec-concl}

We developed several new techniques for proving upper and lower bounds
on the PT height of languages constructed by closing w.r.t.\ the
subword ordering or its inverse. We also considered related
constructions like taking minimal elements, or taking the image by the
incomparability relation.
In general, the PT height of upward closures is bounded with the length of
minimal words.
For downward closures, we developed techniques for expressing them
with D-products and bounding their lengths.
We illustrated these techniques with regular and context-free
languages but more classes can be considered~\cite{zetzsche2015b}.
More importantly, the closures of PT languages have PT height bounded
polynomially in terms of the PT height of the argument.
Our main tool here is the small-subword theorem
that provides tight lower bounds on the PT height of finite
languages, with {ad hoc} developments for $I(L)$.

These results are used to bound the complexity of the two-variable
logic of subwords but we believe that the PT hierarchy can be used
more generally as an effective measure of descriptive complexity. (The
same can be said of the hierarchies of locally-testable languages, or
of dot-depth-one languages).

This research program raises many interesting
questions, such as connecting PT height and other measures, narrowing
the gaps remaining in our Theorems~\ref{thm-ptl-for-upL},
\ref{thm-ptl-for-downL}, and \ref{thm-I-of-PT-is-PT}, and enriching the known
collection of PT preserving operations.


\section*{Acknowledgments}
We thank the anonymous reviewers for their many helpful suggestions.




\appendix
\section{Bounding $f_k(n)$}

From \Cref{sec-PT-levels}, recall the definition of $f_k:\Nat\to\Nat$
where $k$ is a strictly positive integer:
\begin{xalignat}{2}
\tag{\ref{eq-def-f1}}
f_1(n) &= n \:, \\
\tag{\ref{eq-def-fk}}
f_{k+1}(n) &= \max_{0\leq m\leq n}m f_k(n+1-m)+m+f_k(n-m)\:.
\end{xalignat}
We note that each $f_k$ function is monotonic: this is clear for
$f_1$, and Eq.~\eqref{eq-def-fk} guarantees that $f_{k+1}$ is monotonic if
$f_k$ is.

In this appendix we prove the bound on $f_k(n)$ claimed in
\Cref{sec-PT-levels}:
\begin{align}
\tag{\ref{eq-bound-fkn}}
f_k(n) \leq \Bigl(\frac{n+2k-1}{k}\Bigr)^k -1 < \Bigl(\frac{n}{k}+2\Bigr)^k\:.
\end{align}

To prove Eq.~\eqref{eq-bound-fkn} we introduce the following auxiliary
functions, where $0<k\in\Nat$ and $x,y\in\Reals$:
\begin{align}
\tag{$F$}
\label{eq-def-F}
F_k(x) &\egdef \left(\frac{x+2k-1}{k}\right)^k
\:,
\\
\tag{$G$}
\label{eq-def-G}
G_{k,x}(y) & \egdef (y+1)F_k(x-y+1) = \frac{(y+1)(x-y+2k)^k}{k^k}
\:.
\end{align}
Let us check that $G_{k,x}\bigl(\frac{k+x}{k+1}\bigr)=F_{k+1}(x)$ for
any $k>0$ and $x\geq 0$:
\begin{equation}
\tag{$\ast$}
\begin{aligned}
\label{eq-Gkx}
G_{k,x}\left(\frac{k+x}{k+1}\right)
=& \left( \frac{k+x}{k+1} + 1 \right) \frac{1}{k^k}\left(x -
\frac{k+x}{k+1} + 2k \right)^k 
\\
=& \frac{x+2k+1}{k+1}\:\frac{1}{k^k}\left(\frac{kx + 2k^2 +
k}{k+1}\right)^k 
\\
=&
\frac{x+2k+1}{k+1}\:\frac{1}{k^k}\left(\frac{k}{k+1}\right)^k\left(x +
2k + 1\right)^k
\\
=& \left(\frac{x+2k+1}{k+1}\right)^{k+1} = F_{k+1}(x) \:.
\end{aligned}
\end{equation}

\begin{lemma}
\label{lem-max-G}
$F_{k+1}(x)\geq G_{k,x}(y)$ for all $y\in[0,x]$.
\end{lemma}
\begin{proof}
$G_{k,x}$ is well-defined and differentiable
over $\Reals$, its derivative is
\begin{align*}
G_{k,x}'(y) &= \frac{(x-y+2k)^k - (y+1)k(x-y+2k)^{k-1}}{k^k} \\
&= \frac{(x-y+2k)^{k-1}}{k^k}\bigl( (x-y+2k) - (y+1)k \bigr) \\
&= \frac{(x-y+2k)^{k-1}}{k^k}\bigl(x + k - y (k+1) \bigr)
\:.
\end{align*}
Thus $G'_{k,x}(y)$ is $0$ for $y=y_{\max}
\egdef\frac{k+x}{k+1}$, is strictly positive for
$0\leq y < y_{\max}$, and strictly negative for $y_{\max}<y\leq x$. Hence,
over $[0,x]$, $G_{k,x}$ reaches its maximum at $\frac{k+x}{k+1}$
and \eqref{eq-Gkx} concludes the proof.
\end{proof}

\begin{proposition}
$f_k(n)+1\leq F_k(n)$ for all $k,n\in\Nat$ with $k>0$.
\end{proposition}
\begin{proof}
By induction on $k$. For the base case $k=1$, one has
$f_1(n)+1=n+1=F_1(n)$ by combining Eqs.~\eqref{eq-def-f1} and~\eqref{eq-def-F}.  For the inductive case $k\geq 2$, we know by
Eq.~\eqref{eq-def-fk} that
\begin{align*}
f_{k}(n)+1 &=
m\cdot f_{k-1}(n+1-m)+m+f_{k-1}(n-m)+1
\\
\shortintertext{for some $m\in\{0,\ldots,n\},$}
&\leq (m+1) \bigl[f_{k-1}(n+1-m)+1\bigr]
\\
\shortintertext{by monotonicity of $f_{k-1}$,}
&\leq (m+1) F_{k-1}(n+1-m)
\\
\shortintertext{by induction hypothesis,}
&= G_{k-1,n}(m)\leq F_k(n)
\end{align*}
by Eq.~\eqref{eq-def-G} and
\Cref{lem-max-G}.
\end{proof}
This entails $f_k(n)\leq F_k(n)-1$ which is exactly our original
claim.




\bibliographystyle{alpha}
\bibliography{subwords}

\end{document}